\documentclass[final,5p,times,twocolumn]{elsarticle}

\usepackage{graphicx}
\usepackage{amssymb,amsmath,mathrsfs,mathtools}
\usepackage{booktabs}
\usepackage{tabularray}
\usepackage{longtable}
\usepackage{multirow}
\usepackage{array}
\usepackage{bigints}
\usepackage{xurl}
\usepackage{witharrows}
\usepackage{amsthm}
\usepackage[T1]{fontenc}
\usepackage{lmodern}
\usepackage[utf8]{inputenc}
\usepackage{ifthen}
\usepackage{subcaption}
\usepackage[font=scriptsize,labelfont=bf]{caption}
\usepackage[colorlinks=true,linkcolor=blue,citecolor=blue,urlcolor=blue]{hyperref}
\usepackage{algorithm}
\usepackage{algpseudocode}

\newtheorem*{remark}{Remark}
\newtheorem{theorem}{Theorem}[section]

\newtheorem{pr}[theorem]{Proposition}

\newtheorem{cor}[theorem]{Corollary}

\newtheorem{assum}{Assumption}

\allowdisplaybreaks

\begin{document}

\begin{frontmatter}
\title{Physics‐Informed Neural Network Surrogates with Polynomial Chaos Based Uncertainty Propagation for Stochastic Model Predictive Control}

\author[mit]{Srimanta Santra}
\author[mit]{Romir Patel}
\author[mit]{Saikat Mukherjee}
\author[uw]{Steven L. Brunton}
\author[mit]{Richard D. Braatz\corref{cor1}}
\cortext[cor1]{Corresponding author. Email: braatz@mit.edu}

\address[mit]{Massachusetts Institute of Technology, Cambridge, MA 02139, USA}
\address[uw]{Department of Mechanical Engineering, University of Washington, Seattle, WA 98195, USA}

\begin{abstract}
Stochastic partial differential equations (PDEs) govern critical engineering and geophysical systems but are challenging to use for real-time control under parametric uncertainty. We present a unified framework that couples Physics-Informed Neural Networks (PINNs)
with Polynomial Chaos Expansion (PCE) to construct a fast and differentiable surrogate.
The PCE representation enables analytical propagation of parametric uncertainty and
computation of the corresponding moments without requiring Monte Carlo sampling.
We provide an error decomposition for the PINN--PCE surrogate that separates PCE truncation, stochastic quadrature, and PINN approximation errors. Embedding this surrogate into a stochastic model predictive control (SMPC) scheme enables finite-horizon control updates based on analytic mean and covariance predictions. We further show how the surrogate approximation error can be incorporated into tightened probabilistic constraints.
The approach is validated on three benchmarks: the Korteweg–de Vries, Burgers’, and two-dimensional incompressible Navier–Stokes equations capturing dispersive, fully convective and convective-diffusive dynamics. Across all cases, the surrogate achieves real-time controller updates while maintaining prescribed risk levels, matching high-fidelity solvers at a fraction of the computational cost.
\end{abstract}

\begin{keyword}
Physics-Informed Neural Networks, Polynomial Chaos Expansion, Stochastic Model Predictive Control, Uncertainty Quantification, Stochastic Partial Differential Equations, Real-Time Control
\end{keyword}

\end{frontmatter}

\section{Introduction}\label{sec:intro}

Nonlinear partial differential equations (PDEs) underpin many critical engineering and geophysical applications such as coastal solitons, microfluidic mixing, turbulence regularization, and aerodynamic flow control \citep{drazin1989solitons,ghia1982lid}. Spatial discretization schemes such as finite difference, finite volume, or finite element methods can accurately capture these dynamics but demand fine space-time grids and computationally intensive calculations. For instance, a single high-resolution ocean simulation can take hours to days on modern supercomputers \citep{GneitingRaftery2005}, making ensemble runs for uncertainty quantification computationally prohibitive. This has motivated the development of physics-informed surrogate models that replicate PDE dynamics with significantly reduced computational cost.

Physics-Informed Neural Networks (PINNs) address this challenge by embedding PDE residuals directly into the loss function of a neural network, ensuring physical consistency while learning from data \citep{Lu2021DeepXDE}. PINNs have achieved orders-of-magnitude speedups in forward simulations of equations such as Burgers', reaction-diffusion, and Navier-Stokes, and have enabled inverse problems, such as identifying unknown parameters from sparse data \citep{Bertaglia2022,Bertaglia2023,Jin2022}.

Real-world PDEs often involve uncertainties in parameters, boundary conditions, or forcing terms, transforming deterministic PDEs into high-dimensional stochastic PDEs \citep{lord2014stochastic}. Traditional uncertainty quantification methods such as Monte Carlo sampling are computationally expensive due to the need for numerous PDE solves \citep{robert2005monte}. Polynomial Chaos Expansion (PCE) provides an efficient alternative by projecting random inputs onto an orthogonal polynomial basis, yielding closed-form expressions for statistical moments without extensive sampling \citep{xiu2002wiener,lemaitre2010spectral}. By integrating PCE with PINNs, we create a differentiable surrogate that propagates uncertainty analytically, achieving significant computational savings over sampling-based approaches.

The motivation for this work lies in enabling real-time, uncertainty-aware control for systems governed by stochastic PDEs, particularly in stochastic model predictive control (SMPC). SMPC requires solving constrained optimization problems at each time step while enforcing probabilistic constraints on uncertain states and inputs, a task rendered intractable by the computational cost of repeated PDE solves in scenario-based approaches \citep{mesbah2016stochastic}. A PINN-PCE surrogate overcomes this limitation by providing analytic moment predictions, enabling efficient SMPC for real-time applications.

We propose a unified \emph{PINN--PCE} surrogate framework that integrates PINNs with PCE and uses the resulting analytic moments inside an SMPC scheme. The framework is validated on three canonical benchmarks: the one-dimensional Korteweg-de Vries (KdV) equation, the viscous Burgers' equation, and the two-dimensional incompressible Navier-Stokes equations \citep{drazin1989solitons,whitham1974waves,stone2004microfluidic}. These benchmarks are chosen for their relevance to applications like coastal solitons, shock wave dynamics, and microfluidic mixing. While PINNs have been applied to individual PDEs \citep{raissi2019physics}, their integration with PCE for closed-loop SMPC addresses the critical need for computationally efficient uncertainty quantification in control.

Our primary contributions are:
\begin{enumerate}
    \item A robust methodology for training a PINN-PCE surrogate that captures nonlinear PDE dynamics and propagates parametric uncertainty analytically, avoiding Monte Carlo sampling inside the online MPC loop.
    \item An SMPC implementation that uses the PINN--PCE surrogate's analytic moments to enforce probabilistic state and input constraints in real time. The SMPC formulation follows standard mean--covariance MPC principles, while the main contribution is the construction of the prediction model directly from the PINN--PCE coefficients.
    \item Validation of the stochastic-parameter propagation on the KdV and Burgers' benchmarks with 30--100 times computational speedup compared to Monte Carlo-based SMPC, and a scalability demonstration of the SMPC design on the 2D incompressible Navier--Stokes equations.
\end{enumerate}

We establish theoretical error bounds for the surrogate’s accuracy and demonstrate its efficacy through three case studies. The paper is organized as follows: Section~\ref{sec:prob} formalizes the stochastic PDE models, Section~\ref{sec:pinn_pce_algo} details the PINN-PCE surrogate, Section~\ref{sec:smpc} presents the SMPC formulation, Section~\ref{sec:case_study} reports numerical results, and Section~\ref{sec:conclusion} concludes.

\section{Problem Formulation}
\label{sec:prob}
Let $\Omega \subset \mathbb{R}^d$ 
represent a spatial domain with boundary $\partial\Omega$, 
$t \in [0,T]$
a time horizon and $(\Xi,\mathcal{F},\mu)$ a probability space describing parametric uncertainty. Consider the stochastic PDE
\begin{alignat}{2}
\label{eq:pde_general}
&\mathcal{L}\bigl(w(\boldsymbol{x},t,\xi)\bigr)=f(\boldsymbol{x},t,\xi),
  \quad &&(\boldsymbol{x},t,\xi)\in\Omega\times[0,T]\times\Xi, \nonumber\\
  &w(\boldsymbol{x},0,\xi) = w_0(\boldsymbol{x},\xi),\nonumber\\
  &w(\boldsymbol{x},t,\xi) = b(\boldsymbol{x},t,\xi),
  \quad &&\boldsymbol{x}\in\partial\Omega
\end{alignat}
Here, $w$ denotes the PDE state, $b$ denotes the boundary data, and $f$ is a source term.
Assume that $w$ belongs to the Bochner space
$L^2_\mu(\Xi; L^2(0,T; H^s(\Omega)))$, for $s > 0$.
Here, $H^s(\Omega)$ is the Sobolev space of order $s$ on $\Omega$, and $L^2_\mu(\Xi) = L^2(\Xi, \mu)$ is the space of square integrable functions with respect to the probability measure $\mu$. This assumption guarantees finite $L^{2}$ energy in the physical domain and sufficient smoothness in the stochastic dimension for the
PCE analysis that follows \cite{xiu2002wiener}.

\begin{table}[t]
\centering
\caption{Notation used throughout the paper.}
\label{tab:notation}
\small
\begin{tabular}{@{}ll@{}}
\toprule
Symbol & Meaning \\
\midrule
\(w(\boldsymbol{x},t,\xi)\) & Stochastic PDE state (generic formulation) \\
\(u(x,t,\xi)\) & PDE state in case studies (alias for \(w\)) \\
\(\boldsymbol{x}\in\Omega\), \(t\in[0,T]\) & Space, time \\
\(\xi\in\Xi\) & Uncertain parameter \\
\(\nu(\xi)\) & Uncertain viscosity or dispersion coefficient \\
\(\Psi_\alpha(\xi)\) & Orthonormal PCE basis function \\
\(c_\alpha(\boldsymbol{x},t)\) & PCE coefficient \\
\(P\) & PCE truncation order \\
\(M\) & Number of uncertain parameters \\
\(N_P=\binom{M+P}{P}\) & Number of retained PCE terms \\
\(N_s\), \(\{\xi^{(s)},\omega_s\}\) & Quadrature nodes and weights \\
\(\theta\in\mathbb{R}^{N_\theta}\) & PINN parameters \\
\(\widehat{w}(\cdot;\theta)\) & Trained PINN evaluation \\
\(\widehat{W}_P\) & PINN--PCE surrogate \\
\(\mathbf{z}_k(\xi)\in\mathbb{R}^n\) & Discretized state at time step \(k\) \\
\(\mathbf{v}_k\in\mathcal{U}\) & Generic SMPC control input \\
\(U(t)\) & Scalar control input (Burgers', KdV) \\
\(f_x(t),f_y(t)\) & Scalar control inputs (Navier--Stokes) \\
\(\mathbf{F}(\mathbf{z}_k,\mathbf{v}_k,\xi)\) & Discrete one-step transition map \\
\(\boldsymbol{\mu}_k,\boldsymbol{\Sigma}_k\) & Predicted mean and covariance \\
\(\varepsilon_P\) & Surrogate error bound (constraint tightening) \\
\(E_p\) & PINN validation error (Assumption~\ref{assume:pinn_accuracy}) \\
\bottomrule
\end{tabular}
\end{table}

\section{PINN with PCE Framework}
\label{sec:pinn_pce_algo}
By training neural network weights using physical laws, PINNs ensure physics-consistent solutions. On the other hand, PCE provides a framework for quantifying and propagating uncertainty. This section introduces a coupled PINN-PCE framework that bridges deterministic PDE-constrained learning with uncertainty quantification, enabling robust and data-efficient modeling.

\subsection{Physics-Informed Neural Network (PINN)}
A physics-informed neural network \cite{raissi2019physics, karniadakis2021physics} approximates the solution of a PDE through a feedforward network $w_{\theta}(\boldsymbol{x},t)$. The parameters $\theta$ are trained such that the loss function of the neural network penalizes violations of the governing PDE, and its initial and boundary conditions.

\subsection{Polynomial Chaos Expansion (PCE)}
Given a stochastic field \(w(\boldsymbol{x},t,\xi)\), a PCE approximates its dependence on the random parameter \(\xi\in\Xi\subset\mathbb{R}^M\) as
\begin{align}
\label{eq:pce}
w(\boldsymbol{x}, t, \xi)
=
\sum_{\alpha \in \mathbb{N}_0^M}
c_\alpha(\boldsymbol{x}, t)\Psi_\alpha(\xi),
\end{align}
where \(\mathbb{N}_0=\{0,1,2,\ldots\}\), \(M\) is the number of independent stochastic variables, and
\(\alpha=(\alpha_1,\ldots,\alpha_M)\) is a multi-index. In computations, the expansion is truncated to the total-degree index set with the highest order of the polynomial $P$, such that
\[
\mathcal{A}_{M,P}
=
\{\alpha\in\mathbb{N}_0^M:\ |\alpha|=\alpha_1+\cdots+\alpha_M\le P\},
\]
whose cardinality is
\begin{equation}
\label{eq:number_of_terms}
N_P = |\mathcal{A}_{M,P}| = \binom{M+P}{P}.
\end{equation}
$\{\Psi_{\alpha}\}$ is the set of multivariate orthonormal polynomial basis functions with respect to the probability measure $\mu$, i.e.,
\begin{alignat}{2}
            \mathbb{E}\bigl[\Psi_{\alpha}\Psi_{\beta}\bigr]
          =\int_{\Xi}\Psi_{\alpha}(\xi)\Psi_{\beta}(\xi)\,\mathrm{d}\mu(\xi)
          =\delta_{\alpha\beta}.
\end{alignat}
Here, $\mathbb{E}[\cdot]$ is the expectation operator and $\delta_{\alpha\beta}$ is the Kronecker delta. The coefficients $c_{\alpha}(\boldsymbol{x},t):=\mathbb{E}\bigl[w(\boldsymbol{x},t,\xi)\Psi_{\alpha}(\xi)\bigr]$ propagate the uncertainty through space and time \cite{xiu2003modeling}.

\subsection{Coupling PINNs and PCE}
\label{sec:pinn_pce_coupling}
For fixed realizations of the random parameter $\xi$, \eqref{eq:pde_general} reduces to a deterministic PDE in $(\boldsymbol{x}, t)$ \cite{ping2024uncertainty}. To approximate the solution $w(\boldsymbol{x}, t, \xi)$ across multiple realizations of $\xi$, we train a PINN which we represent as
\begin{align}
    w_{\theta}\!: (\boldsymbol{x}, t, \xi) \;\longrightarrow\; \widehat{w}(\boldsymbol{x},t,\xi;\theta)
\end{align}
The choice of $\theta$ determines the network’s ability to satisfy realizations of \eqref{eq:pde_general}. The training procedure must choose $\theta^{\star}$ to provide accurate solutions, which are essential for subsequent uncertainty quantification via PCE.

Using $\xi^{(j)}$ to represent the $j^{\text{th}}$ realization of $\xi$, we define the PINN loss
for that realization as
\begin{align}\label{eq:loss_terms}
\mathcal{J}_{\mathrm{PINN}}^{(j)}(\theta)
  &= \lambda_{\mathrm{PDE}} \mathcal{L}_{\mathrm{PDE}}^{(j)}
   + \lambda_{\mathrm{IC}}  \mathcal{L}_{\mathrm{IC}}^{(j)}
   + \lambda_{\mathrm{BC}}  \mathcal{L}_{\mathrm{BC}}^{(j)}, \\
  \mathcal{J}_{\mathrm{PINN}}(\theta)
  &= \frac{1}{N_s} \sum_{j=1}^{N_s} \mathcal{J}_{\mathrm{PINN}}^{(j)}(\theta),  \end{align}
where the loss terms in \eqref{eq:loss_terms} are
\begin{alignat}{2}
\label{eq:pinn_6}
&\mathcal{L}_{\mathrm{PDE}}^{(j)} 
= \frac{1}{N_q} 
  \sum_{i=1}^{N_q} 
  \bigl| \mathcal{L}\widehat{w}(\boldsymbol{x}_{i}, t_{i}, \xi^{(j)}; \theta) 
        - f(\boldsymbol{x}_{i}, t_{i}, \xi^{(j)}) \bigr|^{2},\\
 \label{eq:pinn_7}
&\mathcal{L}_{\mathrm{IC}}^{(j)} 
= \frac{1}{N_{\mathrm{IC}}} 
  \sum_{i = N_q+1}^{N_q+N_{\mathrm{IC}}} 
  \bigl| \widehat{w}(\boldsymbol{x}_{i}, 0, \xi^{(j)}; \theta) 
        - w_{0}(\boldsymbol{x}_{i}, \xi^{(j)}) \bigr|^{2}, \\
\label{eq:pinn_8}
&\mathcal{L}_{\mathrm{BC}}^{(j)} 
= \frac{1}{N_{\mathrm{BC}}} 
  \sum_{i = N_q+N_{\mathrm{IC}}+1}^{N_q+N_{\mathrm{IC}}+N_{\mathrm{BC}}} 
  \bigl| \widehat{w}(\boldsymbol{x}_{i}, t_{i}, \xi^{(j)}; \theta) 
        - b(\boldsymbol{x}_{i}, t_{i}, \xi^{(j)}) \bigr|^{2}.
\end{alignat}
Here, $\{(\boldsymbol{x}_{i}, t_{i})\}$ are the set of collocation points in space-time coordinates sampled via Latin hypercube sampling, $i\in[1,N_q]$ represents points in the interior of the domain, $i\in[N_q+1,N_q+N_{\mathrm{IC}}]$ represents points such that $t_i=0$, and $i\in[N_q+N_{\mathrm{IC}}+1, N_q+N_{\mathrm{IC}}+N_{\mathrm{BC}}]$ represents points at the boundary such that $\boldsymbol{x}_i\in\partial\Omega$. 
The weights $\lambda_{\mathrm{PDE}}$, $\lambda_{\mathrm{IC}}$, and $\lambda_{\mathrm{BC}}$ are tuned manually to balance the loss terms \cite{wang2021understanding}. 
In this work, we optimize these weights for specific realizations of the stochastic parameter $\xi^{(j)}$, rather than relying on a single nominal value. We believe that this approach promotes a more robust and balanced training process.
In our simulations, the parameter $\theta$ was trained using the gradient-based optimizer Adam until sufficient convergence to an optimal set $\theta^{\star} = \arg\min_{\theta} \mathcal{J}_{\mathrm{PINN}}(\theta)$, typically determined by a loss threshold or early stopping criteria.

For each case study, the loss-weight vector
$\boldsymbol{\lambda}
=
(\lambda_{\mathrm{PDE}},
 \lambda_{\mathrm{IC}},
 \lambda_{\mathrm{BC}})$
is selected before the final training and then kept fixed throughout
the optimization. In particular, the same weight vector is used for
all stochastic realizations $\{\xi^{(j)}\}_{j=1}^{N_s}$; no
realization-dependent tuning is performed. The network parameters are
subsequently obtained by solving
\begin{align*}
   \theta^\star
=
\operatorname*{arg\,min}_{\theta}
\frac{1}{N_s}
\sum_{j=1}^{N_s}
\mathcal{J}_{\mathrm{PINN}}^{(j)}(\theta). 
\end{align*}
Thus, the loss weights are prescribed hyperparameters, whereas
$\theta$ contains the trainable network parameters. The former are
not modified during the optimization of the latter. 



The stochasticity of the solution is captured by training and evaluating the PINN across multiple
realizations of the random parameter $\xi$.
After training, the PINN provides pointwise evaluations
$\widehat{w}(\boldsymbol{x}, t, \xi^{(j)}; \theta^{\star})$
at quadrature nodes $\{\xi^{(j)}\}_{j=1}^{N_s}$.
To construct a stochastic surrogate, we compute the PCE coefficients
$\{\widehat{c}_{\alpha}(\boldsymbol{x}, t;\theta^{*})\}_{|\alpha| \leq P}$ using
non-intrusive spectral projection~\cite{lutjens2021pce},
\begin{align}\label{eq:projection}
\widehat{c}_\alpha(\boldsymbol{x},t;\theta^{*})
\approx \sum_{j=1}^{N_s}
\omega_j \,\widehat{w}(\boldsymbol{x},t,\xi^{(j)};\theta^{\star})\,\Psi_\alpha(\xi^{(j)}),
\end{align}
where $\omega_j$ is the quadrature weight corresponding to $\xi^{(j)}$.

The final PINN--PCE surrogate is
\begin{equation}\label{eq:pinn_pce_surrogate}
\widehat{w}_{P}(\boldsymbol{x}, t, \xi;\theta^\star)
:= \sum_{|\alpha| \leq P} \widehat{c}_{\alpha}(\boldsymbol{x}, t;\theta^\star)\,\Psi_{\alpha}(\xi).
\end{equation}
While related PINN--PCE surrogates have been studied in~\cite{lutjens2021pce,novak2024physics}, our contribution is using the surrogate~\eqref{eq:pinn_pce_surrogate} to construct the finite-horizon mean and covariance prediction model used inside a chance-constrained SMPC problem.

\subsection{Computational scaling of the PINN--PCE surrogate}
\label{subsec:pinn_pce_scaling}
The computational cost of the proposed surrogate has two parts: the offline construction of the PINN--PCE model and the online propagation of its moments inside SMPC. From~\eqref{eq:number_of_terms}, the PCE representation grows polynomially with \(P\) for fixed \(M\), but combinatorially with \(M\). This is the standard dimensionality limitation of polynomial chaos expansions.

If \(C_{\rm PINN}\) denotes the cost of one forward PINN evaluation on the chosen space-time grid, then the total coefficient construction scales as $\mathcal{O}\!\left(N_s N_P C_{\rm PINN}\right)$. After the coefficients are computed, the online mean and covariance predictions scale as
$\mathcal{O}(n)$ and $\mathcal{O}(n^2N_P)$ respectively,
where \(n\) is the dimension of the state.

The PINN training cost depends on the number of stochastic training realizations, collocation points, automatic differentiation operations, and trainable parameters. For one epoch, the dominant cost scales as
\begin{equation}
C_{\rm train}
=
\mathcal{O}
\left(
N_s
\left(N_q+N_{\rm IC}+N_{\rm BC}\right)
C_{\rm AD}(N_\theta)
\right),
\end{equation}
where \(C_{\rm AD}(N_\theta)\) denotes the cost of evaluating the neural network and the derivatives required in the physics residual. Therefore, the proposed approach is most useful when \(M\) is moderate and repeated online uncertainty propagation is required. In this regime, PINN training and PCE projection are performed offline, while the online SMPC problem uses only the propagated PCE coefficients, mean, and covariance.


\subsection{Error Analysis and Convergence of the PINN--PCE Surrogate}
\label{subsec:error_convergence}
Before using the PINN-PCE framework for optimal control, it is important to understand the different sources of numerical error in \eqref{eq:pinn_pce_surrogate}, and their upper bounds. These errors mainly arise from three different sources. First, the finite-order approximation of the stochastic solution via PCE, which truncates the polynomial basis at a finite order $P$. Second, the neural network approximation of the governing equation is limited by the number of trainable parameters $N_\theta$ and the number of quadrature points $N_q+N_{\text{IC}}+N_{\text{BC}}$. Third, the physics-informed loss may not enforce the PDE perfectly. To rigorously validate the convergence of the PINN-PCE surrogate and quantify these error contributions \cite{yarotsky2017error, babuvska2007stochastic}, we establish Theorem~\ref{lem:pce_truncation}. Before proving this result, we state some assumptions.

\begin{assum}[PINN approximation accuracy]
\label{assume:pinn_accuracy}
Let \(D=\Omega\times[0,T]\) and \(\{\xi^{(j)}\}_{j=1}^{N_s}\) denote the stochastic quadrature nodes. For the trained PINN \(\widehat{w}_P(\boldsymbol{x},t,\xi;\theta^\star)\), there exists a finite validation error $E_p$ such that
\begin{equation}
E_p
=
\max_{1\le j\le N_s}
\left\|
w(\cdot,\xi^{(j)})
-
\widehat{w}_P(\cdot,\xi^{(j)};\theta^\star)
\right\|_{L^2(D)}.
\end{equation}
In the numerical studies, \(E_p\) is estimated using independent validation points and high-fidelity reference solutions. When additional approximation regularity is available, \(E_p\) may decrease with the number of trainable parameters and collocation points, but the SMPC construction only requires the existence of an upper bound \(E_p\).
\end{assum}

\begin{assum}[Well-posedness, stochastic regularity, and quadrature accuracy]
\label{assume:number_2}
For each admissible control input and each \(\xi\in\Xi\), the stochastic PDE \eqref{eq:pde_general} has a unique solution
\[
w\in L^2(D;H^r(\Xi))
\]
for some \(r>0\). Moreover, for every retained basis function \(\Psi_\alpha\), the integrand
\(w(\boldsymbol{x},t,\xi)\Psi_\alpha(\xi)\) is sufficiently regular in \(\xi\) so that the quadrature rule satisfies
\[
\left\|
\int_{\Xi} w\Psi_\alpha\,d\mu
-
\sum_{j=1}^{N_s}\omega_j w(\cdot,\xi^{(j)})\Psi_\alpha(\xi^{(j)})
\right\|_{L^2(D)}
\le
C_{\mathrm{quad}}N_s^{-\kappa}
\]
for some \(C_{\mathrm{quad}}>0\) and \(\kappa>0\). 
\begin{assum}[Regularity of the controlled dynamics]
\label{assume:controlled_dynamics}
For every $\xi\in\Xi$, the controlled one-step map
$\mathbf{F}(\mathbf{z},\mathbf{v},\xi)$ is locally Lipschitz
continuous in $(\mathbf{z},\mathbf{v})$ on the compact
state--input region considered by the SMPC problem.
\end{assum}
\end{assum}
\begin{assum}[Stochastic polynomial regularity]
\label{assume:pce_regular}
Assume that the exact stochastic solution admits the
orthogonal expansion
\[
w(\boldsymbol{x},t,\xi)
=
\sum_{\alpha\in\mathbb{N}_0^M}
c_\alpha(\boldsymbol{x},t)\Psi_\alpha(\xi),
\]
and that, for some \(r>0\),
\[
\|w\|_{\mathcal{H}^r_\Psi(\Xi;L^2(D))}^2
:=
\sum_{\alpha\in\mathbb{N}_0^M}
(1+|\alpha|)^{2r}
\|c_\alpha\|_{L^2(D)}^2
<\infty .
\]
This assumption states that the solution map
\(\xi\mapsto w(\cdot,\cdot,\xi)\) has \(r\)-order stochastic regularity in the
coefficient sense associated with the chosen orthonormal polynomial basis.
For standard Legendre or Hermite polynomial families, this coefficient-space regularity is the spectral analog of weighted Sobolev regularity in the
random variable.
\end{assum}

\begin{theorem}\label{lem:pce_truncation}
Suppose Assumptions~\ref{assume:pinn_accuracy}--\ref{assume:number_2} hold. Let
$w \in L^2(\Omega \times [0,T]; H^r(\Xi))$ for $r>0$ be the exact solution to
\eqref{eq:pde_general}, where $\Omega\subset\mathbb{R}^d$ is bounded and
$\Xi$ is an $M$-dimensional probability space with measure $\mathrm{d}\mu(\xi)$
supporting an orthonormal polynomial basis $\{\Psi_\alpha(\xi)\}_{|\alpha|\le P}$. 
Let $\widehat{w}_P$ be the PINN--PCE surrogate, as shown in \eqref{eq:pinn_pce_surrogate}.
Then there exists a constant $C_r = C(r,\|w\|)>0$, independent of $P,N_s$ and $N_\theta$, such that
\begin{equation}\label{eq:pce_error}
\|w-\widehat{w}_P\|
\le C_r P^{-r}
+\sqrt{2\binom{M+P}{P}}\Big[ C_{\mathrm{quad}}N_s^{-\kappa} + M_\alpha E_p\Big],
\end{equation}
where
$M_\alpha := \max_{1\le j\le N_s} |\omega_j \Psi_\alpha(\xi^{(j)})|$,
$E_p$ is the PINN error bound from Assumption~\ref{assume:pinn_accuracy}$; $
$C_{\mathrm{quad}}$ and $\kappa$ depend on the accuracy of the quadrature and the regularity of the integrand from Assumption~\ref{assume:number_2}.
\end{theorem}

\begin{proof} 
Using the triangle inequality,
\begin{align}
\label{eq:lemma_3_1_pinn_1}
\| w - \widehat{w}_P \| \leq \| w - w_P \| + \| w_P - \widehat{w}_P \|.
\end{align}
The PCE truncation error is
\begin{align}
w - w_P = \sum_{|\alpha| > P} c_\alpha(\boldsymbol{x}, t) \Psi_\alpha(\xi).
\end{align} 
Taking the $L^2$-norm of both sides,
\begin{align}
&\| w - w_P \|^
= \int_{D \times \Xi} \Bigg| \sum_{|\alpha| > P} c_\alpha(\boldsymbol{x}, t) \Psi_\alpha(\xi) \Bigg|^2 \, \mathrm{d}\boldsymbol{x} \, \mathrm{d}t \, \mathrm{d}\mu(\xi).
\end{align} 
Using the orthonormality of the basis functions, $\int_{\Xi} \Psi_\alpha \Psi_\beta \, \mathrm{d}\mu = \delta_{\alpha\beta}$, 
\begin{align}
\label{eq:lemma_3_1_pinn_2}
\int_{\Xi} \Bigg| \sum_{|\alpha| > P} c_\alpha(\boldsymbol{x}, t) \Psi_\alpha(\xi) \Bigg|^2 \, \mathrm{d}\mu(\xi) = \sum_{|\alpha| > P} c_\alpha(\boldsymbol{x}, t)^2.
\end{align}
Thus,
\begin{alignat}{2}
\| w - w_P \|^2 
&= \int_{D} \sum_{|\alpha| > P} c_\alpha(\boldsymbol{x}, t)^2 \, \mathrm{d}\boldsymbol{x} \, \mathrm{d}t \nonumber \\
& = \sum_{|\alpha| > P} \| c_\alpha \|_{L^2(D)}^2.
\end{alignat}
Also, 
\begin{align}
\sum_{|\alpha|>P}\|c_\alpha\|_{L^2(D)}^2
&=
\sum_{|\alpha|>P}
(1+|\alpha|)^{-2r}
(1+|\alpha|)^{2r}
\|c_\alpha\|_{L^2(D)}^2 \nonumber\\
&\le
(1+P)^{-2r}
\sum_{|\alpha|>P}
(1+|\alpha|)^{2r}
\|c_\alpha\|_{L^2(D)}^2 \nonumber\\
&\le
(1+P)^{-2r}
\|w\|_{\mathcal{H}^r_\Psi(\Xi;L^2(D))}^2 ,
\end{align}
where the last line follows from Assumption~\ref{assume:pce_regular}. Writing $C_r:=\|w\|_{\mathcal{H}^r_\Psi(\Xi;L^2(D))}$,
\begin{align}
\|w-w_P\|_{L^2(D\times\Xi)}
\le
C_r(1+P)^{-r}.
\end{align}
For simplicity, this is written as \(C_rP^{-r}\) for \(P\ge1\).
Therefore,
\begin{align}\label{eq:lemma_3_1_pinn_3}
\| w - w_P \| \leq C_r P^{-r}.
\end{align}
The difference between the truncated PCE and the PINN-PCE surrogate in \eqref{eq:lemma_3_1_pinn_1} is given by
\begin{align}
w_P - \widehat{w}_P = \sum_{|\alpha| \leq P} (c_\alpha(\boldsymbol{x}, t) - \widehat{c}_\alpha(\boldsymbol{x}, t)) \Psi_\alpha(\xi).
\label{eq:21}
\end{align}
where $\widehat{c}_\alpha(\boldsymbol{x}, t)$ follows from \eqref{eq:projection}. Taking the $L^2$-norm of \eqref{eq:21},
\begin{align}\label{eq:lemma_3_1_pinn_4}
\| w_P - \widehat{w}_P \|^2 = \sum_{|\alpha| \leq P} \left\| c_\alpha - \widehat{c}_\alpha \right\|_{L^2(D)}^2,
\end{align}
since $\int_{\Xi} \Psi_\alpha^2 \, \mathrm{d}\mu = 1$ for orthonormal bases. Using \eqref{eq:pce} and \eqref{eq:pinn_pce_surrogate},
\begin{align}
\label{eq:coefficient_error}
&c_\alpha - \widehat{c}_\alpha = \int_{\Xi} w\Psi_\alpha \, \mathrm{d}\mu - \frac{1}{N_s} \sum_{j=1}^{N_s} \omega_j\widehat{w}(\cdot, \xi^{(j)}; \theta^*) \Psi_\alpha(\xi^{(j)})
\end{align}
The right-hand side of \eqref{eq:coefficient_error} can be rewritten as
\begin{alignat}{2}
\label{eq:coefficient_error_rearranged}
&\int_{\Xi} w\Psi_\alpha \, \mathrm{d}\mu 
- \frac{1}{N_s}\sum_{j=1}^{N_s} \omega_j w (\cdot,\xi^{(j)})\Psi_\alpha(\xi^{(j)}) \nonumber\\
&+ \frac{1}{N_s}\sum_{j=1}^{N_s} \omega_j \Big[ w(\cdot,\xi^{(j)})-\widehat{w}(\cdot,\xi^{(j)};\theta^{*}) \Big] \Psi_\alpha(\xi^{(j)}).
\end{alignat}
The quadrature error in approximating $w\Psi_\alpha$ with $N_s$ quadrature points is
   \begin{align}
   \label{eq:pinn_pce_error_26}
    &e_{q,\alpha}(\boldsymbol{x}, t) = \int_{\Xi} w \Psi_\alpha \, \mathrm{d}\mu
    - \frac{1}{N_s} \sum_{j=1}^{N_s} \omega_j w(\cdot, \xi^{(j)}) \Psi_\alpha(\xi^{(j)}).
    \end{align}
The PINN approximation error in estimating $w(\cdot,\xi^{(j)})$ is
    \begin{align}
    \label{eq:pinn_pce_error_27}
    e_{p,j}(\boldsymbol{x}, t)
    =   w(\cdot, \xi^{(j)}) - \widehat{w}(\cdot, \xi^{(j)}; \theta^*).
   \end{align}
Using definitions \eqref{eq:pinn_pce_error_26}--\eqref{eq:pinn_pce_error_27},  \eqref{eq:coefficient_error} becomes
\begin{align}
c_\alpha - \widehat{c}_\alpha = e_{q,\alpha} + \frac{1}{N_s}\sum_{j=1}^{N_s} \omega_j e_{p,j}\Psi_\alpha(\xi^{(j)}).
\end{align}
Taking the $L^2$-norm of both sides,
\begin{align}
&\left\| c_\alpha - \widehat{c}_\alpha \right\|_{L^2(D)}^2 \nonumber\\
&= \int_{D} \left|e_{q,\alpha} + \frac{1}{N_s}\sum_{j=1}^{N_s} \omega_j e_{p,j}\Psi_\alpha(\xi^{(j)}) \right|^2 \, \mathrm{d}\boldsymbol{x} \, \mathrm{d}t.
\end{align}
Using $(a + b)^2 \leq 2a^2 + 2b^2$,
\begin{equation}
\left\| c_\alpha - \widehat{c}_\alpha \right\|_{L^2(D)}^2 
\leq 2 \left\| e_{q,\alpha} \right\|_{L^2(D)}^2 + \frac{2}{N_s^2} \left\| \sum_{j=1}^{N_s} \omega_j e_{p,j}\Psi_\alpha(\xi^{(j)})  \right\|_{L^2(D)}^2.
\end{equation}
Using \eqref{eq:lemma_3_1_pinn_4},
\begin{align}
\label{eq:before_sqrt}
&\left\| w_P - \widehat{w}_P \right\|^2 \leq 
2 \sum_{|\alpha| \leq P} \left\| e_{q,\alpha} \right\|_{L^2(D)}^2\nonumber\\
&\qquad \qquad \qquad \  + \frac{2}{N_s^2} \sum_{|\alpha| \leq P} \left\| \sum_{j=1}^{N_s} \omega_j e_{p,j}\Psi_\alpha(\xi^{(j)})  \right\|_{L^2(D)}^2
\end{align}
Writing $\epsilon_q^2 = \sum_{|\alpha| \leq P} \left\| e_{q,\alpha} \right\|_{L^2(D)}^2$ and 
$\epsilon_{p}^2 = \frac{1}{N_s^2}\sum_{|\alpha| \leq P} \left\| \sum_{j=1}^{N_s} \omega_j e_{p,j}\Psi_\alpha(\xi^{(j)})  \right\|_{L^2(D)}^2$, and taking the square root of \eqref{eq:before_sqrt} gives
\begin{equation}
\label{eq:error_polynomial_versus_pinnpoly}
\left\| w_P - \widehat{w}_P \right\| \leq \sqrt{2 \epsilon_q^2 + 2 \epsilon_{p}^2}
\leq \sqrt{2} \left( \epsilon_q + \epsilon_{p}\right).
\end{equation}

Combining \eqref{eq:lemma_3_1_pinn_1}, \eqref{eq:lemma_3_1_pinn_3}, and \eqref{eq:error_polynomial_versus_pinnpoly}, we get
\begin{alignat}{2}
    \label{eq:final_error}
    \left\| w - \widehat{w}_P \right\| \leq C_sP^{-s} + \sqrt{2}(\epsilon_q+\epsilon_p).
\end{alignat}
Under Assumption~\ref{assume:number_2},
$w\Psi_\alpha$ is sufficiently regular for the quadrature estimate to hold \cite{mishraestimates2021}, implying
\begin{alignat}{2}
\label{eq:quadrature_error}
|e_{q,\alpha}(\boldsymbol{x},t)| \le C_{\mathrm{quad}}\,N_s^{-\kappa}, \quad \kappa>0.
\end{alignat}
Since the total number of terms in the PCE is $N_P$, \eqref{eq:quadrature_error} implies that
\begin{alignat}{2}
    \label{eq:total_quadrature_error}
    \epsilon_q \leq C_{\mathrm{quad}} N_P N_s^{-\kappa}
\end{alignat}
From Assumption~\ref{assume:pinn_accuracy}, the PINN approximation error satisfies 
$\lvert e_{p,j}(\boldsymbol{x},t)\rvert \le E_p$.
Let $M_{\alpha}=\max_{1\le j\le N_s}{|\omega_j\Psi_\alpha(\xi^{(j)})|}$. Then, $\epsilon_p$ is bounded by
\begin{alignat}{2}
    \label{eq:total_pinn_error}
    \epsilon_p \leq M_\alpha E_p N_P
\end{alignat}
Combining \eqref{eq:final_error}, \eqref{eq:total_quadrature_error}, and \eqref{eq:total_pinn_error}, and substituting $N_P$ from \eqref{eq:number_of_terms} gives
\begin{alignat}{2}
    \label{eq:total_error_simplified}
    \left\| w - \widehat{w}_P \right\| \leq C_rP^{-r} + \sqrt{2\binom{M+P}{P}}\left[C_{\mathrm{quad}}N_s^{-\kappa} + M_\alpha E_p\right].
\end{alignat}
As $N_s \to \infty$, the quadrature error $\to 0$ and, with improved PINN training, $E_p \to 0$. This completes the proof.
\end{proof}

\begin{remark}[Significance of Assumption~\ref{assume:pce_regular}]
The decay rate \(C_r P^{-r}\) is not a consequence of orthogonality alone. Orthogonality only gives Parseval's identity for the PCE coefficients. The algebraic decay follows from the additional regularity of the solution map. If the parameter-to-solution map is non-smooth, discontinuous, or close to a bifurcation, the rate may deteriorate, and the bound established in \eqref{eq:total_error_simplified} should be interpreted as an assumption rather than an automatic property.
\end{remark}

Here, we briefly comment on the validity of Assumption~\ref{assume:pce_regular} for
each benchmark used in Section~\ref{sec:case_study}. For the viscous Burgers' equation with \(\nu\in\mathbb{U}[0.008,0.012]\), the parameter-to-solution map
\(\nu\mapsto u(\cdot,\cdot,\nu)\) is analytic on the interval \([0.008,0.012]\) for the chosen time horizon, since positive viscosity prevents shock formation and the initial data are smooth. For the KdV equation with
dispersion \(\nu\in\mathbb{U}[0.8,1.2]\), the dispersion term is strictly positive, and the two-soliton initial condition is smooth, so standard well-posedness results for KdV give smooth dependence on \(\nu\). The Navier--Stokes case study in Section~\ref{sec:case3} is run with fixed viscosity and therefore does not satisfy Assumption~\ref{assume:pce_regular}; it is used as a numerical scalability test. In each case, the assumed
regularity index \(r\) is not identified analytically; instead, the algebraic decay of \(\|c_\alpha\|\) is confirmed empirically using numerical results.

The next result, Corollary~\ \ref{cor:pinn_pce_l1_error}, provides error bounds for the PINN-PCE surrogate in $L^1$ norm, building on results from Theorem~\ref{lem:pce_truncation}. This result ensures the reliability of the PINN-PCE method in practical applications, where understanding the tradeoff between computational cost and approximation accuracy is essential.
\begin{cor}[Induced \(L^1\) error bound]
\label{cor:pinn_pce_l1_error}
Suppose the conditions of Theorem~\ref{lem:pce_truncation} hold and
\(\mathcal{Y}=D\times\Xi\) has finite measure. Then
\begin{align}
\label{eq:pinn_pce_l1_error}
\left\| w-\widehat{w}_{P} \right\|_{L^1(\mathcal{Y})}
\le
\sqrt{\mu(\mathcal{Y})}
\left\| w-\widehat{w}_{P} \right\|_{L^2(\mathcal{Y})}.
\end{align}
Consequently, if \(P\to\infty\), \(N_s\to\infty\), and the validation error \(E_p\to 0\) along the trained PINN sequence, then the PINN--PCE surrogate converges in both \(L^2\) and \(L^1\).
\end{cor}

\begin{proof}
The result follows directly from the Cauchy--Schwarz inequality:
\begin{align}
&\|w-\widehat{w}_{P}\|_{L^1(\mathcal{Y})}
=
\int_{\mathcal{Y}} |w-\widehat{w}_{P}|\,d\mu\nonumber\\
&\le
\left(\int_{\mathcal{Y}}1\,d\mu\right)^{1/2}
\left(\int_{\mathcal{Y}}|w-\widehat{w}_{P}|^2\,d\mu\right)^{1/2}.
\end{align}
The convergence statement then follows from Theorem~\ref{lem:pce_truncation}, provided that the computable PINN validation error \(E_p\) tends to zero.
\end{proof}

\section{Stochastic Model Predictive Control}
\label{sec:smpc}
This section explains how the PINN--PCE surrogate is converted into the finite-dimensional moment prediction model used by the SMPC optimizer. The key step is to propagate the PCE coefficients of the spatially discretized state under the controlled dynamics and then compute the predicted mean and covariance from these coefficients.

To formulate the SMPC problem, we discretize \(\Omega\) using a spatial grid with \(n\) points \(\{\boldsymbol{x}_i\}_{i=1}^n\), and time using uniform steps of size $\Delta t$ such that \(t_k=k\Delta t\), \(k=0,1,\ldots,K\). The spatially discretized state at $t=t_k$ is
\[
\mathbf{z}_k(\xi)
=
\begin{bmatrix}
w(\boldsymbol{x}_1,t_k,\xi)&
\cdots&
w(\boldsymbol{x}_n,t_k,\xi)
\end{bmatrix}^{\top}.
\]
The controlled discrete-time dynamics are written as
\begin{equation}
\mathbf{z}_{k+1}(\xi)
=
\mathbf{F}(\mathbf{z}_k(\xi),\mathbf{v}_k,\xi),
\label{eq:controlled_discrete_dynamics}
\end{equation}
where \(\mathbf{v}_k\in\mathcal{U}\) is the control input, and \(\mathbf{F}\) is the one-step transition map obtained from the spatial and temporal discretization of the controlled PDE.

The PINN-PCE surrogate approximates the $i^{\text{th}}$ component of the discrete state \(\mathbf{z}_k(\xi)\) at each time step $k$ as
\begin{align}
\hat{z}_{P,ki}(\xi; \theta) = \sum_{|\alpha| \leq P} c_{\alpha,ki}(\theta) \Psi_\alpha(\xi),
\end{align}
where the coefficients $c_{\alpha,k i}$ are calculated using non-intrusive spectral projection, as illustrated in \eqref{eq:projection}. The state vector is thus approximated as $\hat{\mathbf{z}}_{P,k}(\xi; \theta) = [\hat{z}_{P,k1}(\xi; \theta), \ldots, \hat{z}_{P,kn}(\xi; \theta)]^\top$. Thus, the mean and covariance of the state at time $t_k$ are computed from the PCE coefficients as follows
\begin{align}
\mu_{k|t} &= \mathbb{E}[\hat{\mathbf{z}}_{P,k}(\xi; \theta)] = [c_{0,k0}(\theta), \ldots, c_{0,kn}( \theta)]^\top, \\
\Sigma_{k|t} &= \mathbb{E}[(\hat{\mathbf{z}}_{P,k}(\xi; \theta) - \mu_{k|t})(\hat{\mathbf{z}}_{P,k}(\xi; \theta) - \mu_{k|t})^\top]
\end{align}

For a candidate control input \(\mathbf{v}_{k|t}\), the predicted state is evaluated at the quadrature nodes using the pre-trained PINN surrogate
\begin{alignat}{2}
&\widehat{\mathbf{z}}_{P,k+1|t}(\xi^{(j)};\theta)
=
\mathbf{F}
\left(
\widehat{\mathbf{z}}_{P,k|t}(\xi^{(j)};\theta),
\mathbf{v}_{k|t},
\xi^{(j)}
\right),
\nonumber \\ &j=1,\ldots,N_s .
\end{alignat}
The PCE coefficients of the next predicted state are computed by non-intrusive projection:
\begin{equation}
\widehat{\mathbf{c}}_{\alpha,k+1|t}
=
\sum_{j=1}^{N_s}
\omega_j
\widehat{\mathbf{z}}_{P,k+1|t}(\xi^{(j)};\theta)
\Psi_\alpha(\xi^{(j)}).
\end{equation}

The surrogate error bound from Theorem~\ref{lem:pce_truncation}, denoted by $\varepsilon_P$, is used as an additional tightening margin in the probabilistic state constraints, as established in the following proposition.

 \begin{pr}[Error-tightened chance constraint]
\label{prop:error_tightened_chance}
Suppose the PINN--PCE prediction error satisfies
\[
\|\mathbf{z}_{k+i|t}-\widehat{\mathbf{z}}_{k+i|t}\|
\le
\varepsilon_P,
\quad i=1,\ldots,N.
\]
If the tightened constraint
\[
C_j\boldsymbol{\mu}_{k+i|t}
+
\Phi^{-1}(1-\delta_j)
\sqrt{
C_j\boldsymbol{\Sigma}_{k+i|t}C_j^\top
}
+
\|C_j\|\varepsilon_P
\le
b_j
\]
holds for each constraint row \(C_j\), then the corresponding nominal Gaussian chance constraint is protected against all surrogate errors bounded by \(\varepsilon_P\).
\end{pr}

\begin{proof}
For each constraint row \(C_j\),
\[
|C_j(\mathbf{z}_{k+i|t}-\widehat{\mathbf{z}}_{k+i|t})|
\le
\|C_j\|
\|\mathbf{z}_{k+i|t}-\widehat{\mathbf{z}}_{k+i|t}\|
\le
\|C_j\|\varepsilon_P .
\]
Therefore, adding the margin \(\|C_j\|\varepsilon_P\) to the deterministic reformulation of the chance constraint accounts for the worst-case surrogate prediction error. Hence, if the tightened constraint is satisfied by the PINN--PCE prediction, the corresponding untightened constraint is protected for any prediction error bounded by \(\varepsilon_P\).
\end{proof}

The SMPC problem optimizes the control sequence
\(\{\mathbf{v}_{k+i|t}\}_{i=0}^{N-1}\) over a prediction horizon \(N\), using the predicted PCE coefficients, means, and covariances. The optimization problem is
\begin{subequations}
\label{eq:smpc_pinn_pce_discrete}
\begin{align}
\min_{\{\mathbf{v}_{k+i|t}\}_{i=0}^{N-1}}
J
&=
\sum_{i=0}^{N-1}
\bigl(
\boldsymbol{\mu}_{k+i|t}^{\top}Q\boldsymbol{\mu}_{k+i|t}
+
\mathbf{v}_{k+i|t}^{\top}R\mathbf{v}_{k+i|t}
\nonumber\\
&+
\rho\,{\rm tr}(\boldsymbol{\Sigma}_{k+i|t})
\bigr)
\quad+
\boldsymbol{\mu}_{k+N|t}^{\top}P\boldsymbol{\mu}_{k+N|t},
\label{eq:smpc_cost}
\\
{\rm s.t.}\quad
\widehat{\mathbf{c}}_{\alpha,k+i+1|t}
&=
\sum_{j=1}^{N_s}
\omega_j
\widehat{\mathbf{z}}_{P,k+i+1|t}(\xi^{(j)};\theta)
\Psi_\alpha(\xi^{(j)}),
\label{eq:smpc_coeff_dyn}
\\
\boldsymbol{\mu}_{k+i+1|t}
&=
\widehat{\mathbf{c}}_{0,k+i+1|t},
\label{eq:smpc_mean_dyn}
\\
\boldsymbol{\Sigma}_{k+i+1|t}
&=
\sum_{0<|\alpha|\le P}
\widehat{\mathbf{c}}_{\alpha,k+i+1|t}
\widehat{\mathbf{c}}_{\alpha,k+i+1|t}^{\top},
\label{eq:smpc_cov_dyn}
\\
\mathbf{v}_{k+i|t}
&\in \mathcal{U},
\quad i=0,\ldots,N-1,
\label{eq:smpc_input}
\\
C_j\boldsymbol{\mu}_{k+i|t}
&+
\Phi^{-1}(1-\delta_j)
\sqrt{
C_j\boldsymbol{\Sigma}_{k+i|t}C_j^\top
}
+
\|C_j\|\varepsilon_P
\le b_j,
\nonumber\\
&\hspace{1cm}
j=1,\ldots,q,\quad i=1,\ldots,N,
\label{eq:smpc_chance_tightened}
\\
\boldsymbol{\mu}_{k|t}
&=
\boldsymbol{\mu}_{t},
\quad
\boldsymbol{\Sigma}_{k|t}
=
\boldsymbol{\Sigma}_{t}.
\label{eq:smpc_initial}
\end{align}
\end{subequations}


where $J$ is a quadratic cost with positive-definite weighting matrices $Q \succ 0$, $R \succ 0$, and $P \succ 0$ for the mean state, control input, and terminal cost, respectively. The positive constant $\rho$ penalizes the covariance trace to reduce the effects of uncertainty. The convex input constraint set $\mathcal{U} \subseteq \mathbb{R}^m$ enforces control bounds, e.g., $|\mathbf{v}_{k+i|t}| \leq v_{\text{max}}$. The chance constraint $\mathbb{P}[C \mathbf{z}_{k+i|t} \leq 0] \geq 1 - \delta$, where $C \in \mathbb{R}^{q \times n}$ defines state constraints, is reformulated assuming $\widehat{\mathbf{z}}_{P,k+i|t} \sim \mathcal{N}(\mu_{k+i|t}, \Sigma_{k+i|t})$, with $\Phi^{-1}(1 - \delta)$ as the inverse CDF of the standard normal distribution. The initial conditions $\mu_{k|t} = \mu_t$, $\Sigma_{k|t} = \Sigma_t$ are set based on the state estimate at time $t$.

The Gaussian assumption simplifies the chance constraint but may not hold for nonlinear PDEs. However, the $L^1$-norm bound \eqref{eq:pinn_pce_l1_error} in Corollary~\ref{cor:pinn_pce_l1_error} provides a conservative estimate of the error distribution, supporting robust constraint satisfaction.

\begin{remark}[First-order PCE limit]
If the dependence of the solution on \(\xi\) is exactly affine, then a first-order PCE represents the stochastic solution without truncation error. In that special case,
\[
w(t,\xi)=c_0(t)+\sum_{j=1}^{M}c_j(t)\xi_j,
\]
and the mean and covariance reduce to
\[
\mu(t)=c_0(t),\qquad
\Sigma(t)=\sum_{j=1}^{M}c_j(t)c_j(t)^\top .
\]
For nonlinear stochastic dependence, higher-order terms are generally required, and the truncation error is quantified by Theorem~\ref{lem:pce_truncation}.
\end{remark}




\begin{algorithm}[t]
\caption{PINN--PCE-based SMPC}
\label{alg:pinn_pce_smpc}
\begin{algorithmic}[1]
\State Train the PINN surrogate \(\widehat{w}(\boldsymbol{x},t,\xi;\theta^\star)\).
\State Compute PCE coefficients \(\widehat{c}_{\alpha}(\boldsymbol{x},t)\) by quadrature.
\State Form the spatially discretized surrogate state
\[
\widehat{\mathbf{z}}_k(\xi)
=
\sum_{|\alpha|\le P}
\widehat{\mathbf{c}}_{\alpha,k}\Psi_\alpha(\xi).
\]
\State Compute \(\boldsymbol{\mu}_k=\widehat{\mathbf{c}}_{0,k}\) and
\[
\boldsymbol{\Sigma}_k
=
\sum_{0<|\alpha|\le P}
\widehat{\mathbf{c}}_{\alpha,k}
\widehat{\mathbf{c}}_{\alpha,k}^{\top}.
\]
\For{\(k=0,1,\ldots\)}
    \State For each candidate input sequence, propagate the PCE coefficients using
\[
\widehat{\mathbf{c}}_{\alpha,k+i+1|k}
=
\sum_{s=1}^{N_s}
\omega_s
\mathbf{F}
\left(
\widehat{\mathbf{z}}_{k+i|k}(\xi^{(s)}),
\mathbf{v}_{k+i|k},
\xi^{(s)}
\right)
\Psi_\alpha(\xi^{(s)}).
\]
    \State Compute predicted means and covariances from the propagated coefficients.
    \State Solve the tightened SMPC problem.
    \State Apply the first optimal input \(\mathbf{v}_k^\star\).
    \State Measure or estimate the new state and repeat.
\EndFor
\end{algorithmic}
\end{algorithm}

\section{Case Studies}\label{sec:case_study}
We benchmark the proposed pipeline on three PDEs that span dispersive and convective-diffusive regimes: the 1D viscous Burgers', 1D Korteweg de Vries, and the 2D Navier--Stokes equations.

\begin{table}[t]
\centering
\caption{Scaling parameters used in the numerical studies. Burgers' and KdV exercise the full PINN--PCE--SMPC pipeline with a single uncertain parameter propagated by non-intrusive Gauss--Legendre projection. Navier--Stokes is used as a fixed-parameter scalability test of the SMPC design on 2D dynamics; extension to uncertain viscosity is discussed in Section~\ref{sec:conclusion}.}
\label{tab:scaling_parameters}
\begin{tabular}{lccccc}
\toprule
Case & Uncertain param. & \(M\) & \(P\) & \(N_P=\binom{M+P}{P}\) & \(N_s\) \\
\midrule
Burgers & \(\nu\sim\mathbb{U}[0.008,0.012]\) & 1 & 3 & 4 & 5 \\
KdV & \(\nu\sim\mathbb{U}[0.8,1.2]\) & 1 & 2 & 3 & 5 \\
\bottomrule
\end{tabular}
\end{table}

\subsection{Case Study I: Burgers' equation}
\label{sec:case2}
The viscous Burgers' equation models convective wave propagation and dissipation, with applications in traffic flow and shock waves. When the viscosity is uncertain, solving the equation becomes computationally intensive. The general form of the Burgers' equation \cite{MOWLAVI2023111731} with an input is
\begin{equation}
\label{eq:burgers}
\frac{\partial u}{\partial t} + u \frac{\partial u}{\partial x} = \nu \frac{\partial^2 u}{\partial x^2} + U(t), \quad x \in [0, L], \quad t \in [0, T],
\end{equation}
where $u(x,t)$ is the solution, $\nu \in \mathbb{U}[0.008, 0.012]$ is the uncertain viscosity which follows a uniform distribution, and $U(t)$ is the control input. $L = 4.0$ m, and $T = 5.0$ s represent our domain length and time horizon. The PDE is subject to periodic boundary conditions. In this configuration, and for $U=0$, the analytical solution of \eqref{eq:burgers} is
\begin{alignat}{2}
    \label{eq:burgers_analytical}
    u_{\text{an}} = \frac{2 \nu \pi \exp(-\pi^2\nu(t-5)) \sin(\pi x)}{2 + \exp(-\pi^2\nu(t-5)) \cos(\pi x)}.
\end{alignat}
For our simulations, we used the Gauss-Legendre quadrature rule with $N_s=5$ to sample $\nu$. The PINN consisted of three hidden layers with 60, 60, and 40 neurons respectively, using the $\tanh$ activation function, with $N_q=8000$, $N_{\text{IC}}=120$, and $N_{\text{BC}}=120$. Periodic boundary conditions were enforced via a loss term equating $u$ at $x = 0$ and $x = 4$. For the PCE reconstruction using the PINN solutions, Legendre basis functions of order $P = 3$ were used.

\begin{figure}
    \centering
    \includegraphics[width=
    \linewidth]{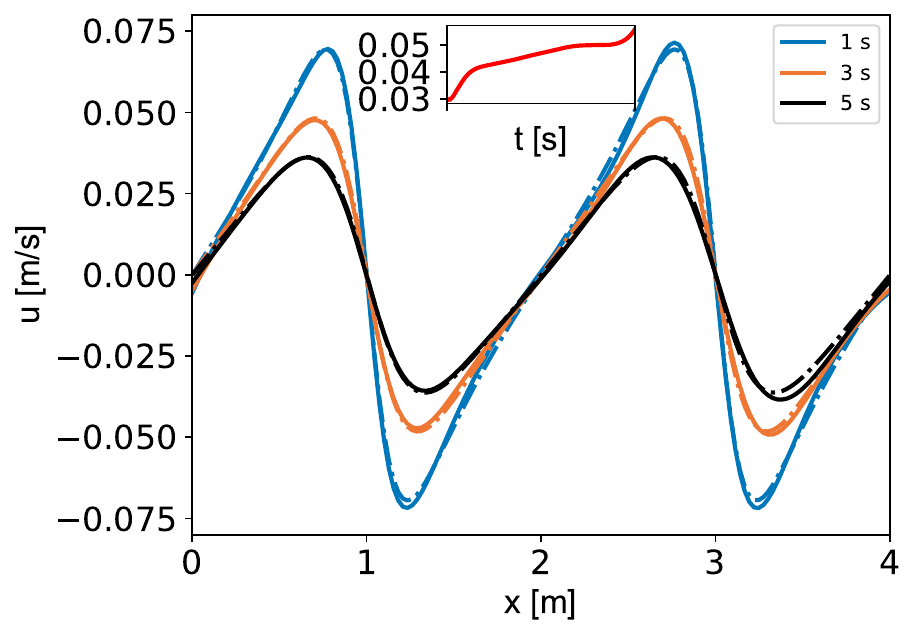}
    
\vspace{-0.3cm}
    \caption{Mean of the PINN-PCE solution (solid line) versus mean of the analytical solution (dotted line) at $t=1$, $t=3$, and $t=5$ s. The inset shows the time evolution of the $L^2$ norm of the relative error, which reaches a maximum of $\sim5.56\%$ at $t=5$ s.}
    \label{fig:Burgers_analysis}
\end{figure}
Figure~\ref{fig:Burgers_analysis} compares the mean of the analytical solution and the PINN-PCE solutions in space at $t=1$, $t=3$, and $t=5$ s. The inset of Figure~\ref{fig:Burgers_analysis} shows the evolution of $|(u-u_{\text{an}})/u_{\text{an}}|_2$ over time. The relative error stays within $6\%$ of the analytical solution for the considered time horizon. 

To empirically verify Assumption~\ref{assume:pce_regular} for this case, we compute the \(L^2(D)\) norm of \(\widehat{c}_\alpha(\boldsymbol{x},t)\) for \(|\alpha|=0,1,\ldots,10\) using a refined \(N_s=15\) Gauss--Legendre quadrature and plot them against \(|\alpha|+1\). Figure~\ref{fig:pce_decay} shows the resulting decay. The coefficient norms drop from \(2.4\times 10^{-1}\) at \(|\alpha|=0\) to \(1.9\times 10^{-6}\) at \(|\alpha|=10\), a decay of nearly five orders of magnitude, and a least-squares log-log fit on the tail \(|\alpha|\ge 2\) gives an empirical rate \(r_{\text{fit}}\approx 6.87\). This confirms the regularity assumed in Theorem~\ref{lem:pce_truncation} on this benchmark, and shows that the operational truncation order \(P=3\) used in Section~\ref{sec:case2} is well inside the regime in which the PCE truncation error is negligible compared to the PINN and quadrature contributions.

\begin{figure}[t]
    \centering
    \includegraphics[width=\linewidth]{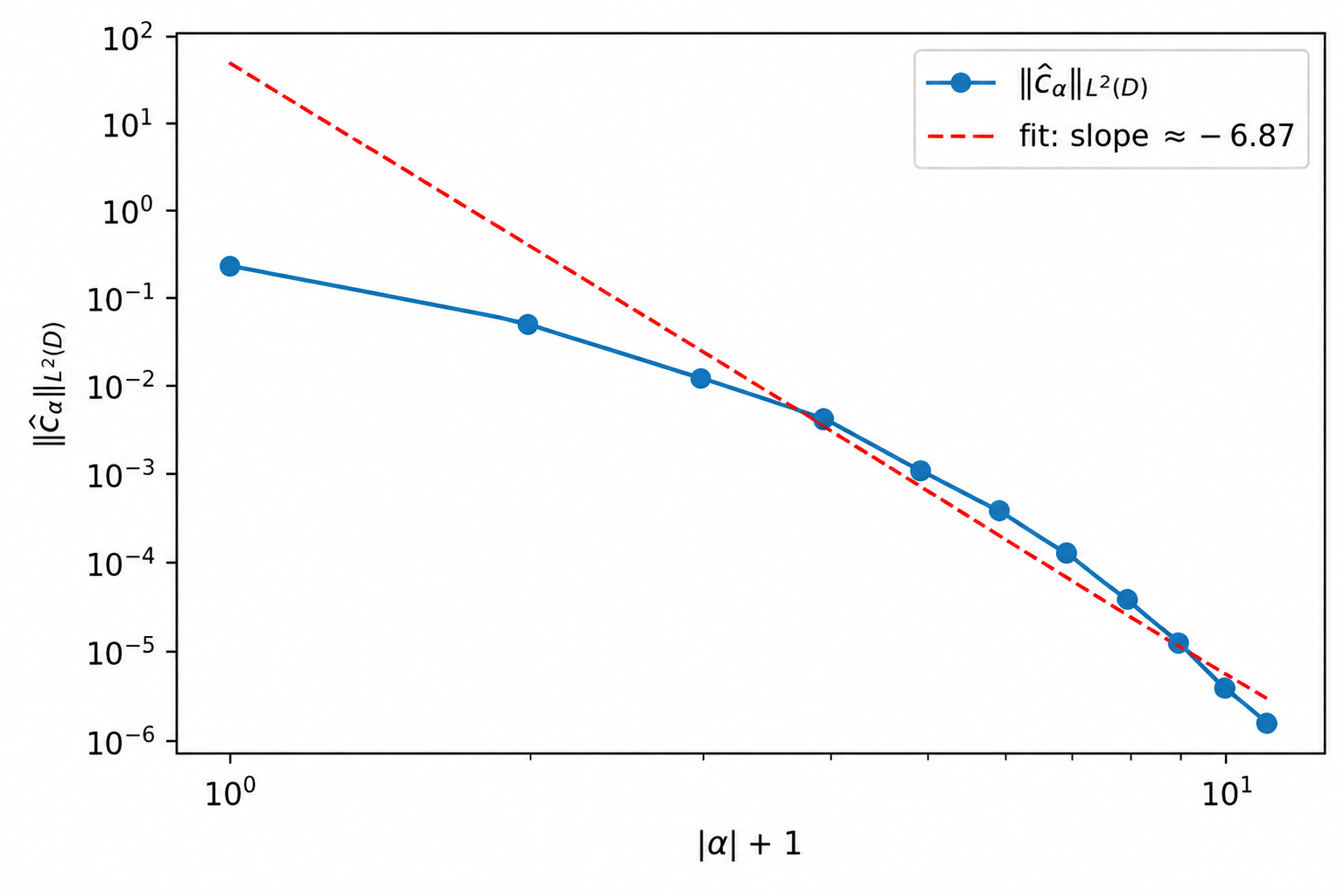}
    \vspace{-0.3cm}
    \caption{Empirical decay of the PCE coefficient norms \(\|\widehat{c}_\alpha\|_{L^2(D)}\) for the Burgers' case study, computed with \(N_s=15\) Gauss--Legendre quadrature. The dashed line is a least-squares fit on the tail \(|\alpha|\ge 2\), giving an empirical algebraic rate \(r_{\text{fit}}\approx 6.87\), consistent with Assumption~\ref{assume:pce_regular}.}
    \label{fig:pce_decay}
\end{figure}

Next, we implement an MPC strategy to track the reference trajectory
\begin{align}
   u_{\text{ref}} = 0.4 \sin\!\frac{3\pi x}{8} \!\left(\!0.9 + 0.3 \cos\!\frac{3\pi t}{50}\right)\!
   + 0.1 \sin\!\frac{3\pi t}{10}. 
\end{align}
The time domain is discretized into $250$ uniform intervals with $\Delta t = 0.02$. After every $4$ time steps, the optimization
\begin{alignat}{2}
\min_{U_{0:N-1}} &\sum_{i=0}^{N-1} \left( \mathbb{E} [\| u_{i|t} - u_{\text{ref}} \|^2] + \lambda U_{i|t}^2 \right) 
\end{alignat}
is solved, where $u_{i|t}$ is the predicted state from the PINN-PCE surrogate, $\lambda = 0.01$ penalizes control inputs, and $N = 100$ is the MPC prediction horizon. Training employed the AdamW optimizer with cosine annealing over 6000 epochs, balancing physics, initial, and boundary condition losses. 
Figure~\ref{fig:Burgers_control} plots the control input $U$ versus time. At $t=1$, $2$, and $3$ s, we also show the mean calculated trajectory along the primary vertical axis, and the relative error $|(u-u_{\text{ref}})/u|$ along the secondary vertical axis. The error remains less than $4\%$. To illustrate the computational speedup, we compare the PINN-PCE-MPC solution to SMPC simulations that use Monte Carlo sampling. For solving \eqref{eq:burgers}, an upwind finite-difference scheme was used with an implicit Euler scheme for time integration. Using $10000$ samples, we observe $\sim25$-fold computational speedup for comparable relative errors. 
\begin{figure}
\centering
\includegraphics[width=\linewidth]{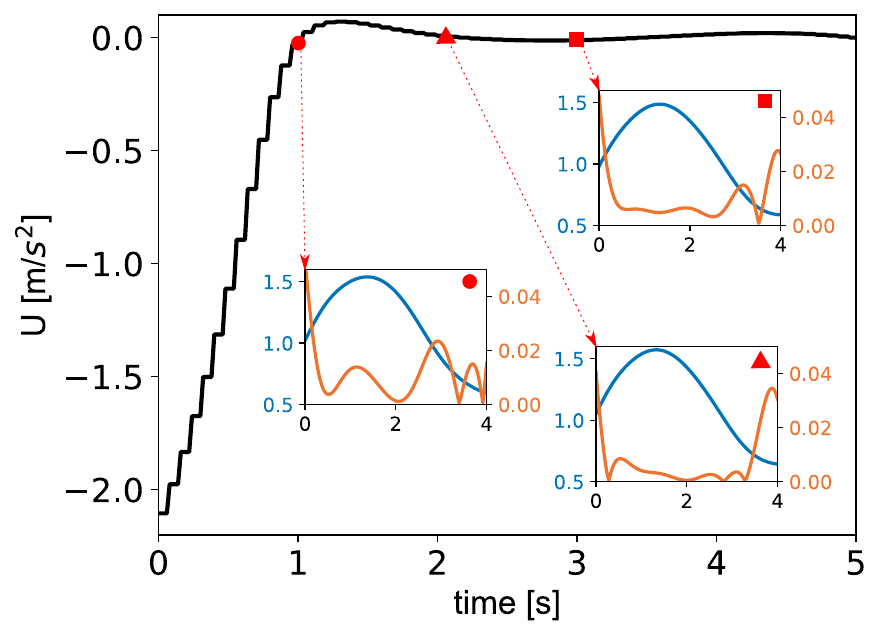}
    \vspace{-0.3cm}
    \caption{Temporal evolution of the control input. The insets at $t=1$, $2$ and $3$ s show the mean of the PINN-PCE solution (plotted on the primary vertical axis in blue) and the relative error between the mean of the PINN-PCE solution and the reference solution (plotted on the secondary vertical axis in orange).}
    \label{fig:Burgers_control}
\end{figure}

\subsection{Case Study II: Korteweg-de Vries equation}
\label{sec:case1}
The Korteweg-de Vries (KdV) equation models nonlinear dispersive wave phenomena, and is commonly applied to study shallow water waves and ion-acoustic waves in plasmas \cite{kdv1895}. Here, we demonstrate the integration of PINNs with PCE to solve the KdV equation when the dispersion coefficient is probabilistically uncertain. Additionally, we use the proposed framework to illustrate how to track a reference trajectory.

This case study is complementary to the Burgers' case study in Section~\ref{sec:case2} rather than a repetition of it. Burgers' equation is second-order in space, and its solution is dominated by convective transport balanced against diffusive dissipation, so uncertainty in $\nu$ primarily affects how sharp fronts smear out. The KdV equation, in contrast, is third-order in space and its solution is governed by a balance between nonlinear steepening and dispersive spreading, so uncertainty in the dispersion coefficient primarily affects the shape and phase speed of coherent structures such as solitons. From the PINN--PCE perspective, this changes the difficulty of the surrogate construction: the physics residual now involves a third-order spatial derivative, which increases automatic-differentiation cost and requires a smoother activation than $\tanh$ to stably represent the third derivative. Here, we use the SiLU activation. From the SMPC perspective, the reference trajectory is tracked against a two-soliton initial condition rather than a smooth wave, thereby exercising the surrogate on a solution class in which small-parameter perturbations produce visible phase shifts. Together, the Burgers' and KdV studies therefore probe convective-diffusive and dispersive regimes with different residual structure, different network design choices, and different sensitivity of the solution to the uncertain parameter, before the framework is applied to the higher-dimensional Navier--Stokes case.
The KdV equation with a control input $U(t)$ is
\begin{alignat}{2}
\label{eq:kdv1}
&\frac{\partial u}{\partial t} + 6 u \frac{\partial u}{\partial x} + \nu(\xi) \frac{\partial^3 u}{\partial x^3} = U(t), \quad x \in [0, L], \quad t \in [0, T],
\end{alignat}
where $u(x,t)$ is the wave amplitude, $\nu \sim \mathbb{U}[0.8, 1.2]$ is the uncertain dispersion coefficient, the spatial domain is $L = 20.0$ m, and the time horizon is $T = 2.0$ s. Periodic boundary conditions are enforced on both sides, i.e., $u(0,t)=u(L,t)$ and $d^i/dx^i u(0,t)=d^i/dx^i u(L,t)$ for $i=1, 2$ or $3$. 
At $t=0$, the initial wave amplitude follows the profile
\begin{align}
   u(x, 0) &= \frac{0.75}{2} \text{sech}^2\!\left[\frac{\sqrt{0.75}}{2}\!\left(x - \frac{L}{3}\right)\right]\nonumber \\
   &\quad + \frac{0.4}{2} \text{sech}^2\!\left[\frac{\sqrt{0.4}}{2}\!\left(x - \frac{2L}{3}\right)\right].
\end{align} 

The PINN model has three hidden layers with 64, 64, and 48 neurons, respectively, using the SiLU activation function, implemented in PyTorch, with $N_q = 12000$, $N_{\text{IC}} = 200$, and $N_{\text{BC}} = 200$. The PCE reconstruction used Legendre basis functions of order $P = 2$ and Gauss-Legendre quadrature with $N_s=5$.

\begin{figure}
    \centering
    \includegraphics[width=
    \linewidth]{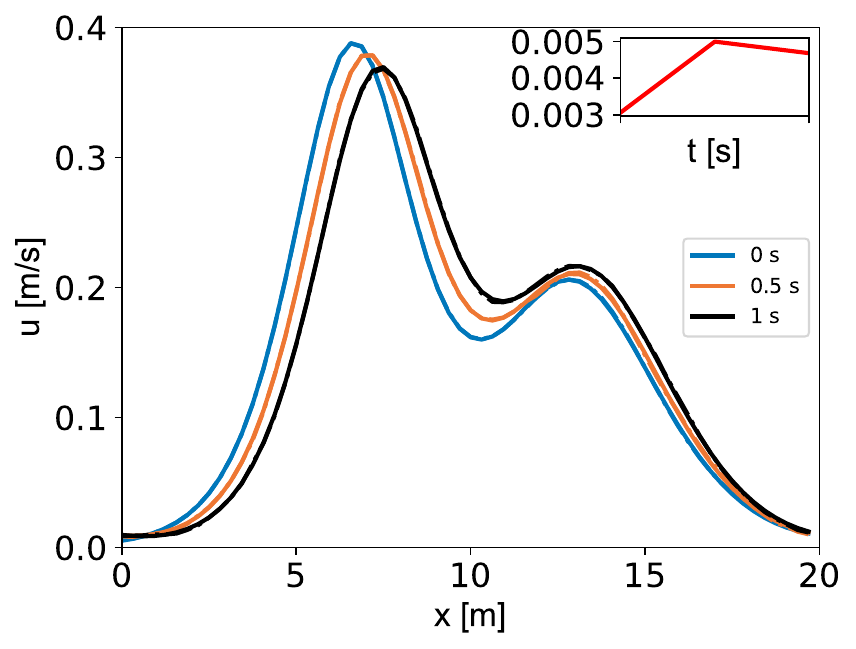}
    
    \vspace{-0.3cm}
    \caption{Mean of the PINN-PCE solution (solid line) versus mean of the analytical solution (dotted line) at $t = 0.5$, $t = 1.0$, and $t = 1.5$ s. The inset shows the time evolution of the $L^2$ norm of the relative error.}
    \label{fig:kdv_analysis}
\end{figure}

The SMPC strategy tracks the reference trajectory
\begin{align}
\label{eq:kdv_ref}
u_{\text{ref}}(x,t) = 0.5 \cos\!\left(\frac{\pi x}{L}\right) e^{-0.1t} + 0.05 \sin\!\left(\frac{2\pi t}{T}\right).
\end{align}
The time domain is discretized into 100 uniform intervals with $\Delta t = 0.02$ s. Every 5 time steps, the optimization
\begin{align}
\label{eq:kdv_mpc}
\min_{U_{0:N-1}} \sum_{i=0}^{N-1} \left( \mathbb{E} [\| u_{i|t} - u_{\text{ref}} \|^2] + \lambda U_{i|t}^2 \right),
\end{align}
is solved, where $u_{i|t}$ is the PINN-PCE predicted state, $\lambda = 0.01$ penalizes control inputs, $N = 20$ is the prediction horizon, and $|U_{i|t}| \leq 5.0$. The optimization used the L-BFGS-B algorithm, with states predicted using a forward-Euler scheme. The framework achieves an average tracking error between $4.0 \times 10^{-4}$ to $7.0 \times 10^{-4}$ over $t \in [0.5, 1.5]$, evaluated at $t = 0.5$, $1.0$, and $1.5$ s.
\begin{figure}
\centering
\includegraphics[width=\linewidth]{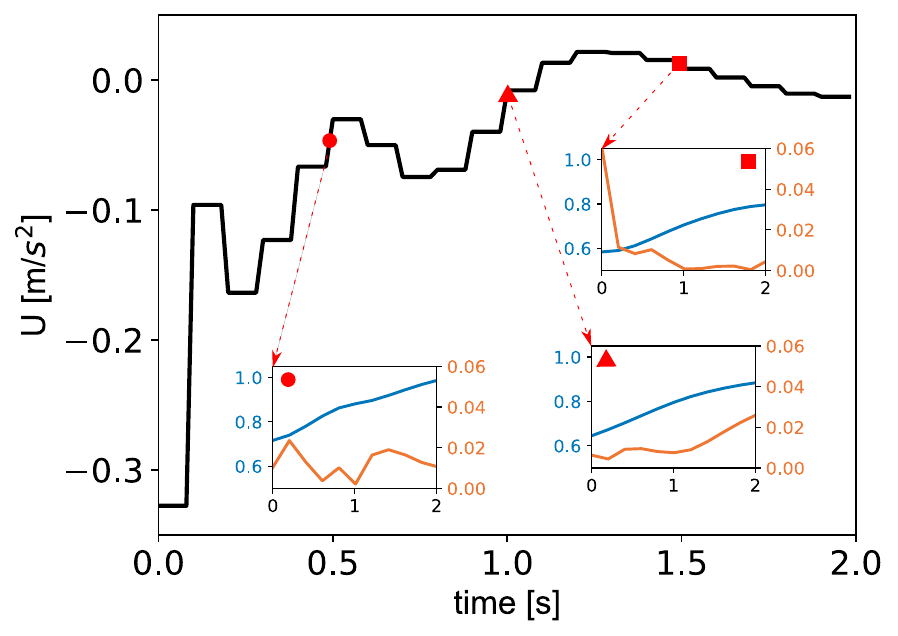}
\vspace{-0.3cm}
\caption{Temporal evaluation of the control input. The insets at $t=0.5$, $1.0$, and $1.5$ s show the mean of the PINN-PCE solution (plotted on the primary vertical axis in blue) and the relative error between the mean of the PINN-PCE solution and the reference solution (plotted on the secondary vertical axis in orange).}
    \label{fig:kdv_control}
\end{figure}

Figure~\ref{fig:kdv_control} illustrates the controlled solution tracking the reference trajectory \eqref{eq:kdv_ref}, the piecewise constant control input, and the relative error on a logarithmic scale. The PINN-PCE-SMPC framework achieves a total execution time of about 550 seconds (150 seconds for PINN training, 250 seconds for controlled PINN training, and 150 seconds for MPC simulation), compared to an estimated 14,000 seconds for a traditional numerical solver with Monte Carlo sampling (10,000 samples).


\subsection{Case Study III: Two-Dimensional Navier--Stokes Equations}
\label{sec:case3}
The final case study considers the two-dimensional Navier--Stokes equations for incompressible flow with uncertain viscosity. Unlike the Burgers'
and KdV examples, the state now contains two velocity components and a
pressure field. We use the Taylor--Green
vortex because its solution for a given viscosity is available in closed form, which provides an independent reference for validating both the
PINN approximation and the stochastic projection.

We consider the nondimensional periodic domain
$\Omega=[0,2\pi]\times[0,2\pi]$,
$t\in[0,5]$, with the governing equations
\begin{align}
\frac{\partial u}{\partial t}
+u\frac{\partial u}{\partial x}
+v\frac{\partial u}{\partial y}
&=
-\frac{\partial p}{\partial x}
+\nu(\xi)
\left(
\frac{\partial^2u}{\partial x^2}
+\frac{\partial^2u}{\partial y^2}
\right)
+f_x(x,y,t),
\label{eq:NS_u}
\\
\frac{\partial v}{\partial t}
+u\frac{\partial v}{\partial x}
+v\frac{\partial v}{\partial y}
&=
-\frac{\partial p}{\partial y}
+\nu(\xi)
\left(
\frac{\partial^2v}{\partial x^2}
+\frac{\partial^2v}{\partial y^2}
\right)
+f_y(x,y,t),
\label{eq:NS_v}
\\
\frac{\partial u}{\partial x}
+\frac{\partial v}{\partial y}
&=0.
\label{eq:NS_cont}
\end{align}
The uncertain kinematic viscosity is
$\nu(\xi)=0.01(1+0.2\xi)$,
$\xi\sim\mathcal U[-1,1]$,
so that
$\nu\in[0.008,0.012]$.
The initial condition is
$u(x,y,0)=-\sin(y)\cos(x)$,\\
$v(x,y,0)=\sin(x)\cos(y)$.
For the uncontrolled problem, $f_x=f_y=0$, the realization-wise
Taylor--Green solution is
\begin{align}\label{eq:ns_u}
u_{\rm an}(x,y,t,\xi)
&=
-A(t,\xi)\sin(y)\cos(x),
\\
\label{eq:ns_v}v_{\rm an}(x,y,t,\xi)
&=
A(t,\xi)\sin(x)\cos(y),
\end{align}
with
$A(t,\xi)=\exp[-2\nu(\xi)t]$,
and corresponding zero-mean pressure
\begin{equation}
p_{\rm an}(x,y,t,\xi)
=
-\frac{A^2(t,\xi)}{4}
\left[
\cos(2x)+\cos(2y)
\right].
\label{eq:NS_exact_p}
\end{equation}
These expressions are used only for independent validation.

The PINN approximates
$(x,y,t,\xi)
\mapsto
(\widehat u,\widehat v,\widehat p)$.
The network contains five hidden layers with 128 neurons per layer and
uses the $\tanh$ activation function. Each training update uses 4096
interior collocation points, 1024 initial-condition points, and 1024
periodic-boundary points. The viscosity uncertainty is represented with
a Legendre PCE of order $P=3$ and $N_s=5$ Gauss--Legendre nodes,
giving $N_P=4$ retained coefficients.

The mean and variance are obtained from
the PCE coefficients as described in Section~\ref{sec:pinn_pce_algo}.
Figure~\ref{fig:NS_forward_mean} shows the mean velocity fields at
$t=1$, $3$, and $5$, while Figure~\ref{fig:NS_forward_std} shows the
corresponding standard-deviation fields.

\begin{figure*}[t]
\centering
\begin{tabular}{ccc}
\includegraphics[width=0.31\textwidth]{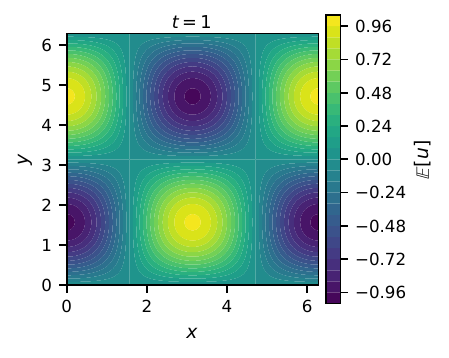} &
\includegraphics[width=0.31\textwidth]{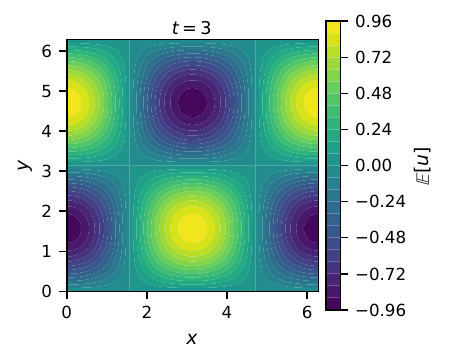} &
\includegraphics[width=0.31\textwidth]{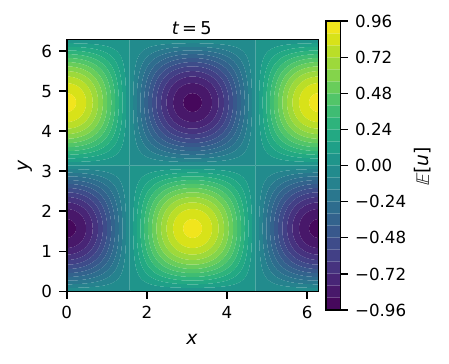}
\\[-0.1cm]
\includegraphics[width=0.31\textwidth]{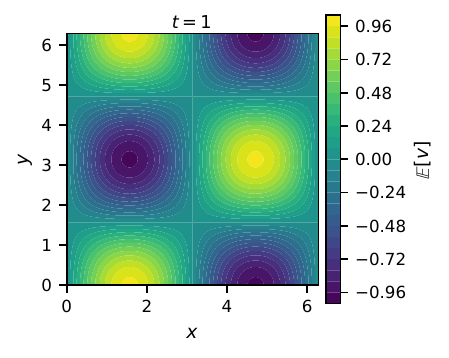} &
\includegraphics[width=0.31\textwidth]{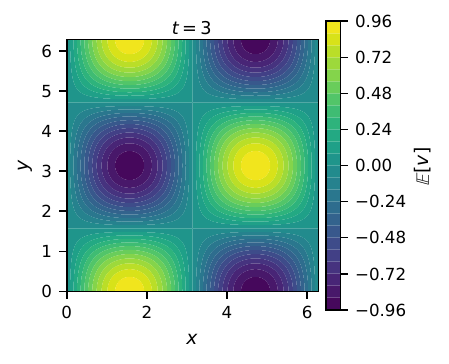} &
\includegraphics[width=0.31\textwidth]{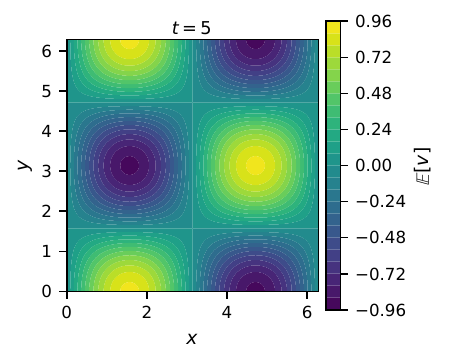}
\end{tabular}
\caption{Mean velocity fields for the uncertain-viscosity
Navier--Stokes case. The top row shows $\mathbb E[u]$ and the bottom
row shows $\mathbb E[v]$ at $t=1$, $3$, and $5$.}
\label{fig:NS_forward_mean}
\end{figure*}

\begin{figure*}[t]
\centering
\begin{tabular}{ccc}
\includegraphics[width=0.31\textwidth]{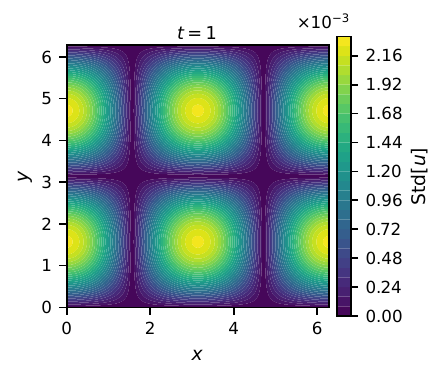} &
\includegraphics[width=0.31\textwidth]{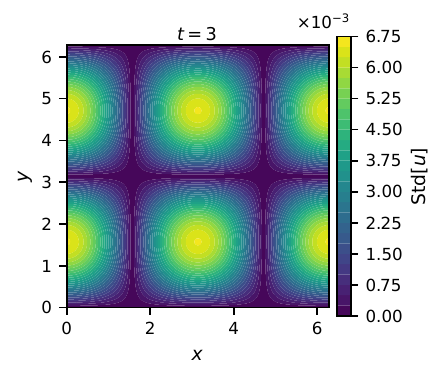} &
\includegraphics[width=0.31\textwidth]{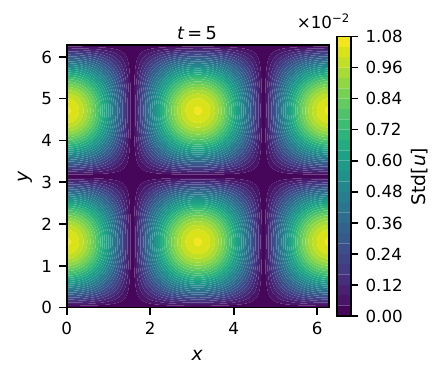}
\\[-0.1cm]
\includegraphics[width=0.31\textwidth]{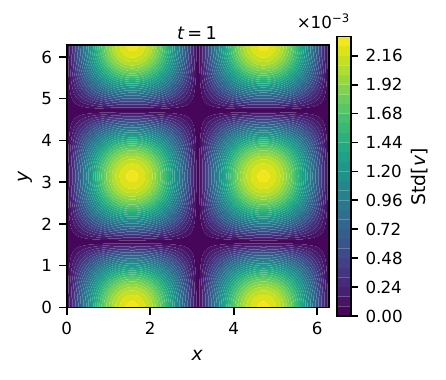} &
\includegraphics[width=0.31\textwidth]{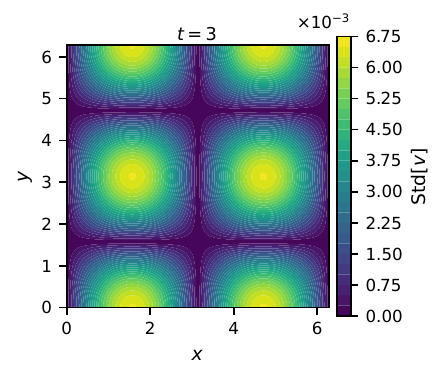} &
\includegraphics[width=0.31\textwidth]{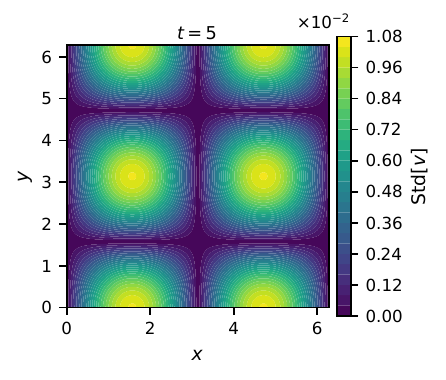}
\end{tabular}
\caption{Standard-deviation fields induced by viscosity uncertainty.
The top row shows $\operatorname{Std}[u]$ and the bottom row shows
$\operatorname{Std}[v]$ at $t=1$, $3$, and $5$.}
\label{fig:NS_forward_std}
\end{figure*}

The stochastic projection was checked against an independent Monte
Carlo calculation with $5\times10^4$ viscosity samples. The relative
PCE vs MC discrepancies over the full time horizon are
$2.02\times10^{-5}$ for the temporal mean and
$1.91\times10^{-4}$ for the temporal standard deviation. These values
measure the PCE projection accuracy, not the PINN approximation error.
The PINN error is evaluated separately using the validation metric
$E_p$ defined in Assumption~\ref{assume:pinn_accuracy}.

\begin{figure*}[t]
\centering
\begin{tabular}{cc}
\includegraphics[width=0.45\textwidth]{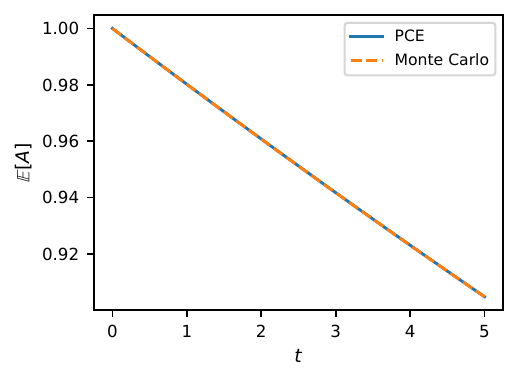} &
\includegraphics[width=0.45\textwidth]{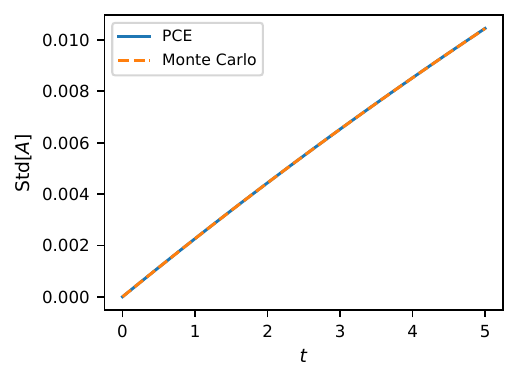}
\\[-0.1cm]
(a) & (b)
\\[0.15cm]
\includegraphics[width=0.45\textwidth]{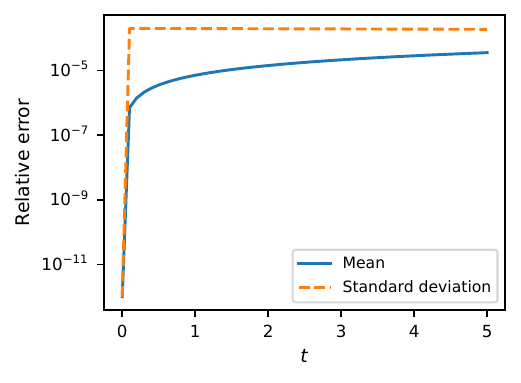} &
\includegraphics[width=0.45\textwidth]{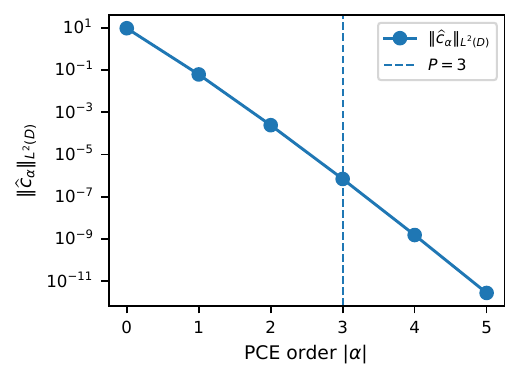}
\\[-0.1cm]
(c) & (d)
\end{tabular}
\caption{Stochastic validation of the Navier--Stokes case.
(a) Mean amplitude from PCE and Monte Carlo.
(b) Standard deviation from PCE and Monte Carlo.
(c) Relative PCE-versus-Monte-Carlo discrepancies.
(d) Decay of the PCE coefficient norms, with the operational order
$P=3$ indicated by the dashed line.}
\label{fig:NS_pce_validation}
\end{figure*}

For the controlled problem, we use the divergence-free forcing
\begin{align}
f_x(x,y,t)
&=
-a(t)\sin(y)\cos(x),
\label{eq:NS_fx}
\\
f_y(x,y,t)
&=
a(t)\sin(x)\cos(y),
\label{eq:NS_fy}
\end{align}
where $a(t)$ is the scalar manipulated input. This forcing preserves
the Taylor--Green spatial structure, so the controlled velocity field
can be same way \eqref{eq:ns_u} and \eqref{eq:ns_v}
with amplitude dynamics
\begin{equation}
\dot A(t,\xi)
=
-2\nu(\xi)A(t,\xi)+a(t).
\label{eq:NS_amplitude}
\end{equation}

The time-varying reference amplitude is
\begin{equation}
A_{\rm ref}(t)
=
0.95
+
0.05
\cos\left(
\frac{2\pi t}{T}
\right),
\qquad T=5,
\label{eq:NS_Aref}
\end{equation}
which decreases from $1$ at $t=0$ to $0.90$ at $t=2.5$ and returns to
$1$ at $t=5$. The corresponding reference field is
\begin{align}
u_{\rm ref}(x,y,t)
&=
-A_{\rm ref}(t)\sin(y)\cos(x),
\\
v_{\rm ref}(x,y,t)
&=
A_{\rm ref}(t)\sin(x)\cos(y).
\label{eq:NS_reference}
\end{align}

For this Taylor--Green family,
\begin{equation}
\|u-u_{\rm ref}\|_{L^2(\Omega)}^2
+
\|v-v_{\rm ref}\|_{L^2(\Omega)}^2
=
2\pi^2
(A-A_{\rm ref})^2,
\label{eq:NS_spatial_amp_equiv}
\end{equation}
so amplitude tracking is equivalent, up to a fixed scaling, to
tracking the complete two-dimensional velocity field.

The case-specific SMPC problem is
\begin{align}
\min_{\{a_{k+i|k}\}_{i=0}^{N_c-1}}
\quad
\sum_{i=1}^{N_p}
&\Bigg[
\mathbb E
\left[
\left(
A_{k+i|k}
-
A_{\rm ref}(t_{k+i})
\right)^2
\right]
\nonumber\\
&+
\rho
\operatorname{Var}(A_{k+i|k})
+
\lambda_a
a_{k+\min(i-1,N_c-1)|k}^{2}
\Bigg],
\label{eq:NS_SMPC_cost}
\end{align}
subject to
\begin{equation}
-0.10
\le
a_{k+i|k}
\le
0.10.
\label{eq:NS_control_constraint}
\end{equation}
We use
$\Delta t_c=0.1$,
$N_p=15$,
$N_c=5$,
with
$\rho=0.2$,
$\lambda_a=0.005$.
Only the first optimal control move is applied before the horizon is
shifted.
Figure~\ref{fig:NS_control_performance} shows the resulting closed-loop
response. For the independent PCE--SMPC verification calculation, the
uncontrolled mean trajectory has an amplitude RMSE of
$4.51\times10^{-2}$, whereas the controlled mean RMSE is
$2.04\times10^{-4}$. The control input varies between approximately
$-0.0434$ and $0.0815$. The integrated expected tracking error is
reduced from $1.06\times10^{-2}$ to $1.90\times10^{-4}$, corresponding
to a reduction of approximately $98.2\%$.
\begin{figure*}[t]
\centering
\begin{tabular}{ccc}
\includegraphics[width=0.31\textwidth]{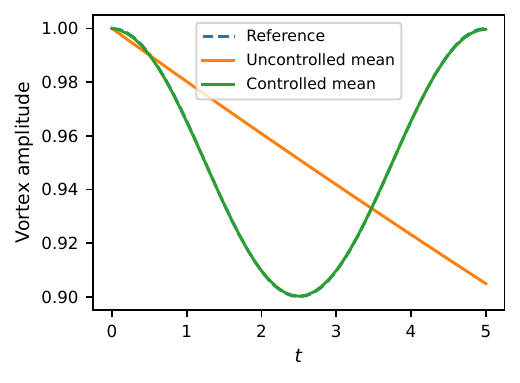} &
\includegraphics[width=0.31\textwidth]{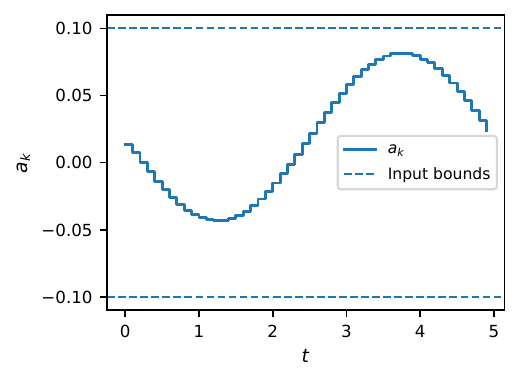} &
\includegraphics[width=0.31\textwidth]{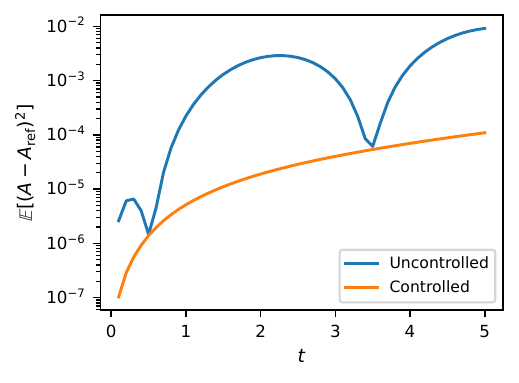}
\\[-0.1cm]
(a) & (b) & (c)
\end{tabular}
\caption{Closed-loop response for the uncertain-viscosity
Navier--Stokes case.
(a) Reference, uncontrolled mean, and controlled mean vortex amplitude.
(b) Receding-horizon control input.
(c) Expected tracking error with and without control.}
\label{fig:NS_control_performance}
\end{figure*}

The controlled mean velocity fields are shown in
Figure~\ref{fig:NS_controlled_mean}. The controller changes the vortex
strength while preserving the spatial phase and symmetry of the
Taylor--Green structure.

\begin{figure*}[t]
\centering
\begin{tabular}{ccc}
\includegraphics[width=0.31\textwidth]{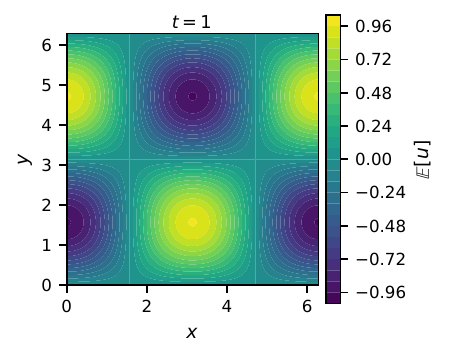} &
\includegraphics[width=0.31\textwidth]{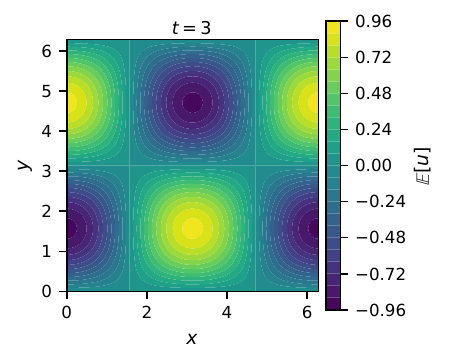} &
\includegraphics[width=0.31\textwidth]{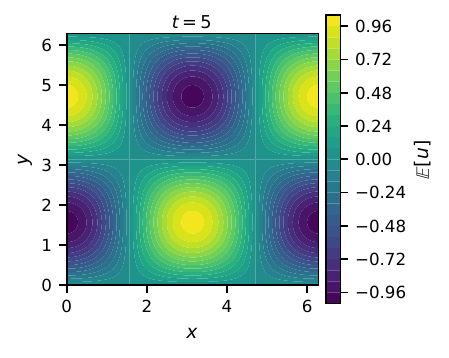}
\\[-0.1cm]
\includegraphics[width=0.31\textwidth]{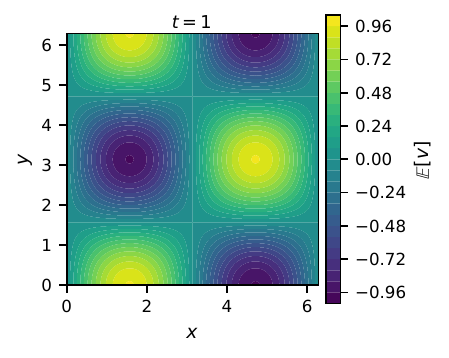} &
\includegraphics[width=0.31\textwidth]{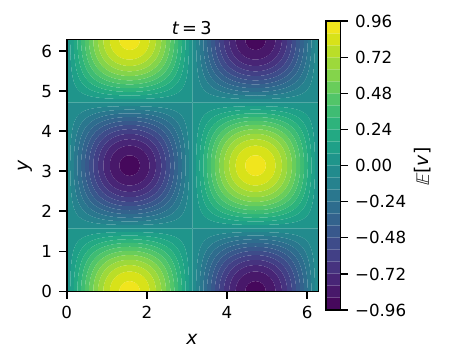} &
\includegraphics[width=0.31\textwidth]{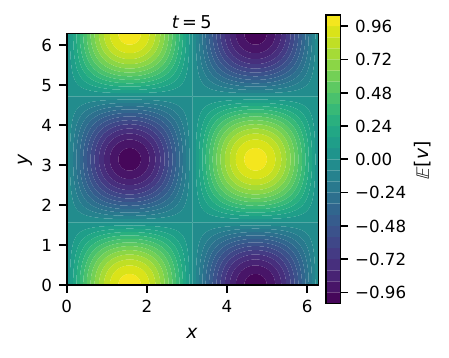}
\end{tabular}
\caption{Controlled mean velocity fields. The top row shows
$\mathbb E[u]$ and the bottom row shows $\mathbb E[v]$ at $t=1$,
$3$, and $5$.}
\label{fig:NS_controlled_mean}
\end{figure*}

Finally, Figure~\ref{fig:NS_tracking_error_fields} reports the absolute
mean tracking-error fields. We use an absolute error rather than a
pointwise relative error because both velocity components contain
spatial zero crossings, where division by the reference velocity can
produce artificially large values.

\begin{figure*}[t]
\centering
\begin{tabular}{ccc}
\includegraphics[width=0.31\textwidth]{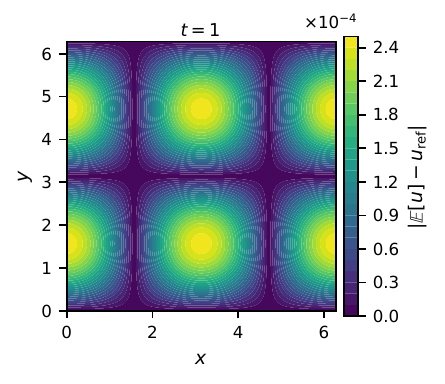} &
\includegraphics[width=0.31\textwidth]{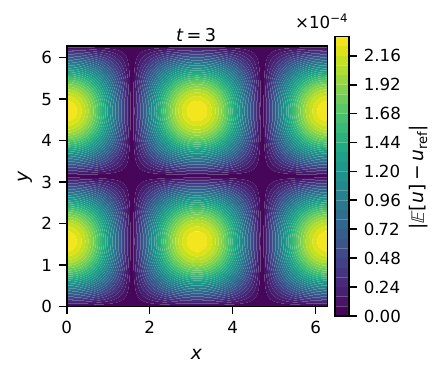} &
\includegraphics[width=0.31\textwidth]{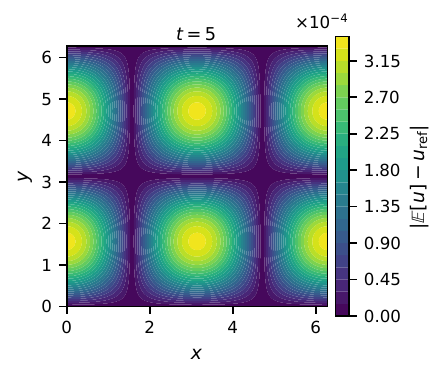}
\\[-0.1cm]
\includegraphics[width=0.31\textwidth]{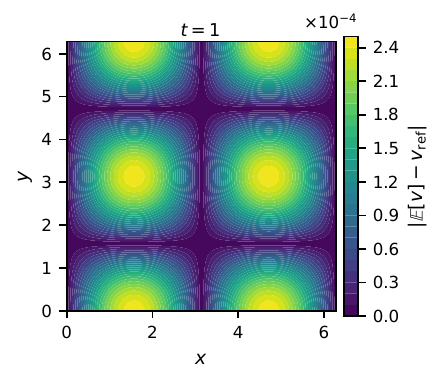} &
\includegraphics[width=0.31\textwidth]{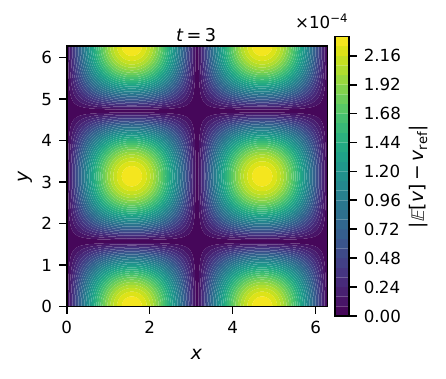} &
\includegraphics[width=0.31\textwidth]{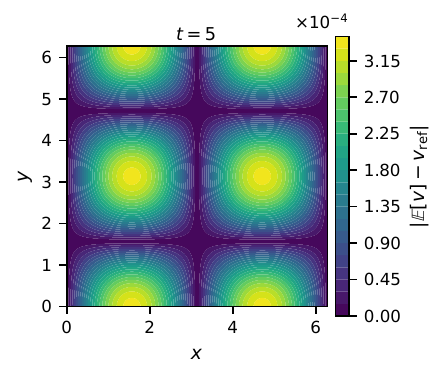}
\end{tabular}
\caption{Absolute mean tracking-error fields. The top row shows
$|\mathbb E[u]-u_{\rm ref}|$ and the bottom row shows
$|\mathbb E[v]-v_{\rm ref}|$ at $t=1$, $3$, and $5$.}
\label{fig:NS_tracking_error_fields}
\end{figure*}

This case study therefore tests three distinct components of the
framework: two-dimensional PINN prediction, PCE-based propagation of
viscosity uncertainty, and receding-horizon control of the uncertain
vortex amplitude. The Taylor--Green structure keeps each component
independently verifiable, so the example is interpreted as a
two-dimensional stochastic control benchmark rather than as a general
turbulent-flow application.

\section{Conclusions}\label{sec:conclusion}
This paper develops a practical surrogate-based route to stochastic MPC for PDE-governed systems with parametric uncertainty. The approach combines a physics-informed neural network for fast state prediction with non-intrusive PCE for moment-based uncertainty propagation, thereby avoiding Monte Carlo sampling inside the online control loop.
On the analysis side, we provide an explicit decomposition of the total error into (i) PCE truncation, (ii) stochastic quadrature or projection, and (iii) PINN approximation terms, yielding convergence bounds in $L^2$ and an induced $L^1$ bound under finite-measure assumptions. These results clarify how accuracy scales with the truncation order $P$, the number of quadrature points $N_s$, and the PINN approximation error.

From a control standpoint, the proposed surrogate enables a practical stochastic MPC implementation for PDE systems under parametric uncertainty. Each control update follows the standard receding horizon loop estimate the plant state, optimize a horizon length tracking policy under constraints, apply the first input, and repeat while replacing expensive sampling based uncertainty propagation with a PINN-PCE pipeline. This separation (PINN for fast horizon prediction, PCE for moment based risk evaluation) explains why the controller maintains tracking performance with orders-of-magnitude lower runtime than MCMC based SMPC in the reported Burgers', KdV, and 2D Navier--Stokes studies.

We emphasize the scope of the theoretical guarantees provided in this work. Theorem~\ref{lem:pce_truncation} and Corollary~\ref{cor:pinn_pce_l1_error} establish an error decomposition and convergence of the PINN--PCE surrogate. Proposition~\ref{prop:error_tightened_chance} establishes a one-step probabilistic constraint guarantee: if the surrogate prediction error is bounded by \(\varepsilon_P\) at a given prediction step, the tightened chance constraint~\eqref{eq:smpc_chance_tightened} protects the corresponding nominal chance constraint against that error. A full recursive feasibility and closed-loop stability analysis for the proposed chance-constrained SMPC under generic nonlinear stochastic PDE dynamics with a data-driven surrogate lies outside the scope of the present paper. Such an analysis would require additional structure, for example a robust terminal set construction, a tube-based reformulation, or a contraction-based argument tailored to the discretized PDE. We therefore restrict the present theoretical claims to the surrogate error decomposition and the one-step tightening result, and leave the closed-loop stability analysis to future work.

\section{Acknowledgment}
 This research was supported by the U.S. Food and Drug Administration under the FDA BAA-22-00123 program, Award Number 75F40122C00200.

\section*{Code availability}
The numerical scripts, processed data, parameter samples, and plotting routines
required to reproduce the final reported figures and tables will be deposited
in a public archive with a persistent identifier before publication. During
review, the complete reproducibility package is available from the
corresponding author upon reasonable request.

\bibliographystyle{elsarticle-num}
\bibliography{reference}
\end{document}